\documentclass[11pt]{article}
\usepackage{amsmath}
\usepackage{amsthm}
\usepackage{graphicx}
\usepackage{txfonts}
\usepackage{enumitem}
\usepackage[colorinlistoftodos]{todonotes}
\usepackage[a4paper, margin=0.7in]{geometry}
\usepackage{amssymb}
\usepackage{apacite}
\usepackage{mathrsfs}
\usepackage{bm}
\usepackage{natbib}
\usepackage{times}
\usepackage{caption}
\usepackage{subcaption}
\usepackage{float}
\newtheorem{theorem}{Theorem}
\newtheorem{assumption}{Assumption}
\numberwithin{theorem}{section}
\newtheorem{Lemma}{Lemma}
\newtheorem{Corollary}{Corollary}
\numberwithin{Lemma}{section}

\newtheorem{remark}{Remark}

\numberwithin{Definition}{section}

\allowdisplaybreaks
\usepackage{algorithm}
\usepackage{algpseudocode}

\usepackage{fancyhdr}
\usepackage{xcolor}
\usepackage[colorlinks = true,
            linkcolor = blue,
            urlcolor  = blue,
            citecolor = blue,
            anchorcolor = blue]{hyperref}
\usepackage{color,soul}
\usepackage{array}
\newcolumntype{L}[1]{>{\centering}m{#1}}
\usepackage{float}
\floatstyle{plaintop}
\restylefloat{table}
\usepackage{multirow}
\usepackage{longtable}
\newcommand{\Ex}{\mathbb{E}}          % expectation
\newcommand{\Var}{\operatorname{Var}} % variance
\newcommand{\iid}{\stackrel{\text{iid}}{\sim}}

\newcommand{\EE}{\mathbb{E}}

\newcommand{\cp}{k^\star}
\newcommand{\tstar}{\tau^\star}
\newcommand{\Keta}{\mathcal{K}_\eta}
\newcommand{\norm}[1]{\left\lVert #1 \right\rVert}
\begin{document}
\title{\large Single Change-Point Detection via Energy Distance with Application to Genomic Data}
%% The left and right page headers are defined here:
\author{\small Suthakaran Ratnasingam $^{\footnote{Corresponding author. Email: suthakaran.ratnasingam@csusb.edu}}$ \\
\small  Department of Mathematics \\ \small California State University, San Bernardino, San Bernardino, CA 92407, USA}
\date{}
\maketitle
\vspace{-1cm}

\maketitle

\begin{abstract}
In this paper, we develop and analyze a nonparametric procedure for detecting a single change point in sequences of independent observations using energy distance. The asymptotic properties of the test statistic are derived under both null and alternative hypotheses. Under the null hypothesis, for any fixed candidate split point, the standardized statistic $\mathcal{Z}_{n,k}$ converges
to a standard normal limit. For global detection, we use the scan statistic
$T_n=\max_{k\in K_\eta}|\mathcal{Z}_{n,k}|$ and calibrate critical values using a permutation test, which yields valid
type I error control under exchangeability. The simulation study shows that the proposed method demonstrates much better robustness across various error distributions. To handle multiple change points in practical applications, the method is combined with a binary segmentation approach. The breast cancer cell line (MDA157) from cDNA microarray CGH data is used to illustrate the detection and estimation capabilities of the proposed method for genomic sequences.
\end{abstract}

\section{Introduction}

Change point detection is a fundamental problem in statistical analysis, where the goal is to identify points in time at which the statistical properties of a sequence of observations change. There is a vast amount of literature available on change point detection. For instance, \cite{Chernoff1964}, \cite{gardner1969}, \cite{Bhattacharya1968}, \cite{sen1975a,sen1975b}, and \cite{Gupta1996} studied have explored changes in location in the normal distribution. Moreover, inference about changes in variance while assuming a constant mean has been investigated by \cite{hsu1977}, \cite{davis1979}, \cite{abraham1984}, and \cite{chen1997}. Further, \cite{srivastava1986} used likelihood ratio tests to detect changes in the mean of a sequence of independent normal random variables. \\

Nonparametric methods have been developed to avoid restrictive distributional assumptions. Examples include rank-based procedures such as the Mann–Whitney $U$-test (\cite{haw1977}), Kolmogorov–Smirnov statistics (\cite{bro1993}), and Pettitt’s test (\cite{petti1979}). These approaches are distribution-free but often limited to detecting location shifts. \cite{gombay1995} applied $U$-statistics to analyze distributional changes, while \cite{gombay2001} studied their properties under alternatives. More recent advances include change-point detection with weighted two-sample $U$-statistics (\cite{DehlingVukWendler2022}) and high-dimensional extensions (\cite{Boniece2024}). These methods highlight the flexibility of  $U$-statistics as a general tool for nonparametric change-point problems. Other notable works include sequential monitoring schemes \cite{kirch2021}, asymptotic delay times for $U$-statistic tests \cite{kirch2022}, and robust extensions for functional data \cite{Wegner2024}. While these $U$-statistic methods excel in flexibility and asymptotic efficiency, they often rely on specific kernels (e.g., quadratic or weighted) tailored to dependence or dimensionality, and few explore general metric spaces or distance-based measures that capture multifaceted distributional changes in location, scale, and shape. Nonparametric change-point detection has evolved significantly beyond classical and $U$-statistic-based procedures. Modern approaches can be broadly categorized into kernel-based, graph-based, and multiscale methods. Kernel-based methods, such as those utilizing the Maximum Mean Discrepancy (MMD) \cite{li2019kernel, gretton2012kernel}, map data into Reproducing Kernel Hilbert Spaces to detect complex distributional changes. Concurrently, the graph-based framework introduced by \cite{chen2017graph} uses minimum spanning trees or nearest-neighbor graphs to construct test statistics sensitive to high-dimensional geometric structures. In the context of segmentation, state-of-the-art methods have moved beyond standard binary segmentation to more robust scanning schemes. These include Wild Binary Segmentation (WBS) introduced by \cite{fryzlewicz2014wild}, Seeded Binary Segmentation (SeedBS) developed by \cite{kovacs2020seeded}, and multiscale approaches such as SMUCE proposed by \cite{frick2014multiscale}, which provide improved localization guarantees. Despite these advances, energy statistics introduced by \cite{sze2013} offer a unique advantage by providing a rigid motion-invariant characterization of equality of distributions without the need for bandwidth selection common in kernel methods or the graph construction steps required in geometric approaches. Although energy statistics have been integrated into hierarchical clustering approaches, such as the E-Divisive method of \cite{matteson2014ecp}, the asymptotic properties of an energy-distance based scan statistic formulated within a $U$-statistic framework have received relatively limited attention. Our goal is to contribute a distance-based, analytically standardized energy scan statistic with an explicit global testing rule that is simple to implement and naturally captures changes in location, scale, and shape.
\\

In parallel, energy statistics were introduced by \cite{sze2000} and \cite{sze2013}, who noted that these statistics can be expressed as $U$- or $V$-statistics constructed from pairwise distances. Energy distance provides a metric between distributions that captures differences in location, scale, and shape, making it particularly robust. While energy statistics have been applied to goodness-of-fit testing and multivariate normality ( \cite{sze2004}; \cite{sze2004}; \cite{sze20051}), their use in change-point detection has not been fully explored. The energy distance is defined to be the statistical distance between probability distributions. According to \cite{sze2013}, the energy distance between two independent $p$-dimensional random variables $X$ and $Y$ and for a fixed power $0<\omega<2$, the (population) energy distance is
\begin{equation}\label{m1}
\begin{aligned}
\mathcal{E}(X,Y) = 2\mathbb{E}\|X - Y\|^\omega - \mathbb{E}\|X - X'\|^\omega - \mathbb{E}\|Y - Y'\|^\omega, 
\end{aligned}
\end{equation}
provided $E\|X\|^{\omega} < \infty$, $E\|Y\|^{\omega} < \infty$, where $X'$ is an i.i.d. copy of $X$ and $Y'$ is an i.i.d. copy of $Y$. The distance between two distributions in terms of expected values of powers of Euclidean distances. When $\omega = 1$, the energy distance defined in (\ref{m1}) becomes $\mathcal{E}(X,Y) = 2\mathbb{E}|X - Y| - \mathbb{E}|X - X'| - \mathbb{E}|Y - Y'|$. It can be shown that $\mathcal{E}(X,Y) \geq 0$ and $\mathcal{E}(X,Y) = 0$ if and only if $X \stackrel{\text{$d$}}{=} Y$. The two-sample energy statistic for data $\mathbf{X}=\left(X_1, X_2, \ldots, X_n\right)$ and $\mathbf{Y}=\left(Y_1, Y_2, \ldots, Y_m\right)$ the data energy distance between $\mathbf{X}$ and $\mathbf{Y}$ is defined as follows:
\begin{equation}\label{m2}
\begin{aligned}
\mathcal{E}_{n, m}(\mathbf{X}, \mathbf{Y})= & \frac{2}{n m} \sum_{i=1}^n \sum_{j=1}^m\left|X_i-Y_j\right|  -\frac{1}{n^2} \sum_{i=1}^n \sum_{j=1}^n\left|X_i-X_j\right|-\frac{1}{m^2} \sum_{i=1}^m \sum_{j=1}^m\left|Y_i-Y_j\right|.
\end{aligned}
\end{equation}

This is an empirical measure of the distance between $X$ and $Y$. The energy statistics finds many different studies have already been performed on this topic. \cite{sze20051}  proposed a test based on the energy distance for multivariate normality. \cite{rizzo2009},  and \cite{rizzo2016} considered one sample goodness-of-fit test for Pareto distributions. For more details, readers are refer to \cite{sze2004}, \cite{sze2007}, \cite{sze20051} and \cite{sze20052}.\\

In this paper, we propose a novel change-point detection procedure based on energy statistics. Although our theoretical development focuses on a single change point, the method can be extended to multiple change points via recursive schemes such as binary segmentation. This makes it well-suited to applications such as genomic sequence analysis, where multiple structural changes commonly occur. Energy-distance ideas have previously been used for change-point analysis. In particular, energy-based procedures such as E-Divisive \cite{matteson2014ecp} are often implemented within the ECP framework: they use the sample energy distance to detect distributional changes and typically calibrate significance via resampling methods (e.g., permutations) within a hierarchical segmentation scheme. The focus of the present paper is different. We study a single-change-point scan statistic constructed from a weighted two-sample $U$-statistic based on the energy distance. Our main contributions are: (i) an explicit studentisation that yields an asymptotically pivotal statistic for interior splits, (ii) consistency of the estimated change-point location, and (iii) a transparent and implementable calibration for the resulting global scan statistic. This framework clarifies how the proposed approach relates to, and differs from, existing energy-based segmentation methods. 
The energy distance provides a robust and flexible measure of discrepancy between distributions, capturing differences in location, scale, and shape through pairwise distances. The proposed energy distance (ED)-based procedure uses these properties to achieve broad applicability and improved robustness relative to existing approaches. More broadly, while energy statistics have a substantial literature in goodness-of-fit testing and dependence assessment, their use in change-point detection has developed along several lines, including energy-based segmentation methods. The present work contributes a complementary perspective by focusing on a single-split $U$-statistic scan with explicit asymptotic analysis and calibration.\\

The rest of the paper is organized as follows. In Section 2, we present the main results of the test statistic, and the corresponding asymptotic results are established. Simulations with various settings are conducted in Section 3 to investigate the performance of the proposed methods. In Section 4, we provide a real data application, extending the single change point method via binary segmentation to detect multiple change points in genomic data, illustrating the detection and estimation process.  Some discussion is presented in Section 5.

\section{Main Results}

Let $X_1,\dots,X_n\in\mathbb R^p$ be independent. Fix $0<\omega <2$ (throughout we take $\omega =1$),
define the symmetric degree-2 kernel $h(x,y):=\|x-y\|^\omega $. Let $X,X' \stackrel{\text{i.i.d.}}{\sim}F$ and set $\mu=\Ex\,h(X,X')$.
For a split index $k\in\{1,\dots,n-1\}$, define the two segments $X_{1:k}=(X_1,\dots,X_k)$ and $X_{k+1:n}=(X_{k+1},\dots,X_n)$. We study the single change-point problem. Under the null hypothesis, all observations follow the same distribution. That is,
\[
H_0:\quad X_1,\ldots,X_n \stackrel{\text{i.i.d.}}{\sim} F.
\]
Under the alternative, there exists an unknown location $k^\star \in \{1,\dots,n-1\}$ such that
\[
H_1:\quad X_1,\ldots,X_{k^\star} \stackrel{\text{i.i.d.}}{\sim} F_1, 
\quad 
X_{k^\star+1},\ldots,X_n \stackrel{\text{i.i.d.}}{\sim} F_2,
\qquad F_1 \neq F_2.
\]

The objective is to test whether a change occurs and, if so, to estimate the change-point location $k^\star$. Define the degree-2 $U$-statistics
\[
U_{XY}=\frac{1}{k(n-k)}\sum_{i=1}^{k}\sum_{j=k+1}^{n}h(X_i,X_j),\quad
U_{XX}=\frac{2}{k(k-1)}\sum_{1\le i<\ell\le k}h(X_i,X_\ell),\quad
U_{YY}=\frac{2}{(n-k)(n-k-1)}\sum_{k< j<\ell\le n}h(X_j,X_\ell).
\]
Using the equation (\ref{m2}), the energy distance between two samples  $X_k$ and $Y_{n-k}$ is defined as:
\begin{equation}\label{m3}
\begin{aligned}
\mathcal{E}(X_k, Y_{n-k}) = 2U_{XY}-U_{XX}-U_{YY} =: E_k .
\end{aligned}
\end{equation}

We propose the test statistics to detect the structural change, 
\begin{equation}\label{m}
\begin{aligned}
\mathcal{Z}_{n,k} = \frac{k(n-k)}{\sqrt{2}\,n\,s_{\Psi,n}}\,E_k,
\end{aligned}
\end{equation}

where $s_{\Psi,n}$ estimates the scale of the degenerate kernel appearing in the Hoeffding decomposition. Let
\begin{equation}\label{dec26}
\begin{aligned}
\bar h_{i\cdot} &:= \frac{1}{n-1}\sum_{j\neq i} h(X_i,X_j), \qquad
\bar h_{\cdot\cdot} := \frac{2}{n(n-1)}\sum_{1\le i<j\le n} h(X_i,X_j),\\
\hat\Psi_{ij} &:= h(X_i,X_j) - \bar h_{i\cdot} - \bar h_{j\cdot} + \bar h_{\cdot\cdot}, \qquad 1\le i<j\le n,\\
s_{\Psi,n}^2 &:= \frac{2}{n(n-1)}\sum_{1\le i<j\le n}\hat\Psi_{ij}^2.
\end{aligned}
\end{equation}

\begin{proof}
Since $\hat{\Psi}_{ij}$ is the usual double-centered estimator of the degenerate kernel, $s_{\Psi,n}^2$ is a (pooled) $V$-statistic of order 2 with kernel $\Psi(\cdot,\cdot)^2$. Square-integrability in Assumption 3 yields the $V$-statistic LLN, hence the stated convergence.
\end{proof}

This studentisation is chosen to (i) be stable and common across all candidate splits $k\in K_\eta$, enabling a uniform-in-$k$ limit theory, and (ii) target the variance of the degenerate kernel that drives the asymptotic fluctuation of the weighted $U$-statistic under $H_0$. In finite samples under $H_1$, using an all-pairs scale estimator can absorb part of the signal through between-segment variability; the degenerate-centering above reduces this effect while keeping the statistic asymptotically pivotal under $H_0$. The test statistic $\mathcal{Z}_{n,k}$ measures the standardized energy distance between the two subsequences split at point $k$. A large value of $|\mathcal{Z}_{n,k}|$ suggests a significant difference between the distributions of the subsequences. This suggests a potential change point at $k$. To detect a change point, we scan over interior set $\mathcal K_\eta:= \bigl\{k\in\{1,\dots,n-1\}:\ \lceil \eta n\rceil \le k \le \lfloor (1-\eta)n\rfloor,\ k\ge 2,\ n-k\ge 2\bigr\}$ with fixed $\eta\in(0,1/2)$ and estimate the change-point by
\begin{equation}\label{m}
\begin{aligned}
\hat k \;=\; \arg\max_{k\in\mathcal K_\eta}\,|\mathcal{Z}_{n,k}|.
\end{aligned}
\end{equation}
To establish the asymptotic distribution of the test statistic $\mathcal{Z}_{n,k}$, we make the following assumptions.

\begin{assumption}\label{A:sampling}
(i) Under $H_0$, $X_1,\dots,X_n$ are i.i.d.\ from a distribution $F$ on $\mathbb{R}^d$ with Euclidean norm $\|\cdot\|$. 
(ii) Under $H_1$, the sequence $(X_i)_{i=1}^n$ is independent with a single change point $k^\star \in\{1,\dots,n-1\}$: $X_1,\dots,X_{k^\star} \iid F_1$ and $X_{k^\star+1},\dots,X_n \iid F_2$ with $F_1\neq F_2$. 
(iii) Candidate splits are restricted to the interior index set $\mathcal K_\eta=\{\lceil\eta n\rceil,\dots,\lfloor(1-\eta)n\rfloor\}$ for fixed $\eta\in(0,1/2)$, and $k^\star/n\to\tau^\star\in(\eta,1-\eta)$. Here $\tau$ denotes the change-point location limit $k^\star/n\to\tau^\star$.
\end{assumption}

\begin{assumption}\label{A:moment}
Let $h(x,y):=\|x-y\|$. Assume $\mathbb E\|X\|^2<\infty$ under $H_0$ and $H_1$, hence
$\sigma^2:=\operatorname{Var}\{h(X,X')\}<\infty$ for $X,X'\stackrel{iid}{\sim}F$. Under $H_1$, the analogous second moments and variances are finite for the mixture $F_\lambda:=\lambda  F_1+(1-\lambda )F_2$, where $\lambda \in[0,1]$ denotes a generic mixture weight.
\end{assumption}

\begin{assumption}\label{A:Ustat}
Let $\mu:=\mathbb{E}\,h(X,X')$ and define the first-order projection
$\psi(x):=\mathbb{E}[h(x,X')]-\mu$. Define the degenerate (second-order) kernel $\Psi(x,y):=h(x,y)-\mu-\psi(x)-\psi(y)$, so that the Hoeffding decomposition holds:
$h(x,y)=\mu+\psi(x)+\psi(y)+\Psi(x,y)$, with
$\mathbb{E}[\Psi(x,X')]=0$ for all $x$.
Assume $\psi\in L^2(F)$ and $\Psi\in L^2(F\otimes F)$, and that the remainder terms in the Hoeffding
decomposition of $U_{XY},U_{XX},U_{YY}$ are $o_p(n^{-1/2})$ uniformly over $k/n\in[\eta,1-\eta]$.
\end{assumption}

\begin{assumption}\label{A:delta}
Under $H_1$, set
$\mu_{11}=\Ex\|X-X'\|$ for $X,X'\sim F_1$,
$\mu_{22}=\Ex\|Y-Y'\|$ for $Y,Y'\sim F_2$,
$\mu_{12}=\Ex\|X-Y\|$ for $X\sim F_1,Y\sim F_2$,
and define the population energy distance
$\Delta:=2\mu_{12}-\mu_{11}-\mu_{22}>0$.
\end{assumption}

Assumptions \ref{A:sampling}–\ref{A:delta} formalize the single–change-point setting needed for our asymptotic results. 
Assumption~\ref{A:sampling} specifies an i.i.d.\ baseline under $H_0$ and an independent single-change model under $H_1$, while restricting candidate splits to the interior set $\mathcal K_\eta$ so that both segments have order-$n$ length; this interior trimming is standard and enables uniform-in-$k$ limits for $U$-process change-point statistics \citep{cs1997,DehlingVukWendler2022}. 
Assumption~\ref{A:moment} imposes finite second moments for the Euclidean distance kernel $h(x,y)=\|x-y\|$, ensuring that the pooled $V$-statistic $s_{\Psi,n}^2$ and the constituent $U$-statistics satisfy LLN/CLT properties both under $H_0$ and under the mixture laws $F_\lambda=\lambda F_1+(1-\lambda)F_2$ that arise in $H_1$ \citep{sze2013,RizzoSzekely2016}. 
Assumption~\ref{A:Ustat} collects the standard $U$-statistic regularity conditions: Hoeffding decompositions with square-integrable first-order projections and degenerate remainders of order $o_p(n^{-1/2})$ uniformly over $k/n\in[\eta,1-\eta]$; in particular, the first-order terms cancel in $E_k$ under $H_0$, which is key for the null limit \citep{hoeff1948,DehlingVukWendler2022}. 
Finally, Assumption~\ref{A:delta} posits a fixed, strictly positive energy-distance contrast $\Delta$ between $F_1$ and $F_2$, supplying the population signal that yields divergence of $\mathcal{Z}_{n,k}$ at the true split and localization of the maximizer under $H_1$ \citep{sze2013}. 
With these conditions in place, we establish the null limit and fixed-alternative behavior of the standardized statistic $\mathcal{Z}_{n,k}$ below.\\

First, we employ the following symmetric-weight representation (all sums over $1\le i<j\le n$). For $m=n-k$, define the symmetric weights
\[
w_{n,k}(i,j)=
\begin{cases}
-\dfrac{m}{k-1}, & 1\le i\ne j\le k,\\[6pt]
\ \ 1, & (i\le k<j)\ \text{or}\ (j\le k<i),\\[6pt]
-\dfrac{k}{m-1}, & k<i\ne j\le n,
\end{cases}
\qquad
a_{n,k}(i,j):=2\,w_{n,k}(i,j),
\]
and note the identities
\begin{equation}\label{eq:rowsums}
\sum_{j\ne i} w_{n,k}(i,j)=0\quad (1\le i\le n),\qquad
\sum_{1\le i<j\le n} a_{n,k}(i,j) = 0.
\end{equation}

\begin{Lemma}\label{lem:weights}
For all $k$, define $S_{n,k}:=\dfrac{k(n-k)}{n}\,E_k$,
\begin{equation}\label{eq:symm-id}
S_{n,k}=\frac{1}{n}\sum_{1\le i<j\le n} a_{n,k}(i,j)\,h(X_i,X_j).
\end{equation}
and, uniformly for $\gamma:=k/n\in[\eta,1-\eta]$,
\begin{equation}\label{eq:A(lambda)}
\frac{1}{n^2}\sum_{1\le i<j\le n} a_{n,k}(i,j)^2 \ \longrightarrow\ A(\gamma):=2.
\end{equation}
Moreover,
\begin{equation}\label{eq:row-l2}
\max_{1\le i\le n}\ \frac{1}{n^2}\sum_{j\ne i} a_{n,k}(i,j)^2 \ \le\ \frac{C}{n}
\quad\text{for some }C<\infty\ \text{independent of }n,k\ \text{when }\gamma\in[\eta,1-\eta].
\end{equation}
\end{Lemma}

\begin{proof}
Expanding $E_k=2U_{XY}-U_{XX}-U_{YY}$ gives \eqref{eq:symm-id}.
For \eqref{eq:A(lambda)}:
\[
\sum_{i<j} a_{n,k}(i,j)^2
=\binom{k}{2}\Big(\frac{2m}{k-1}\Big)^2 + (km)\cdot 2^2 + \binom{m}{2}\Big(\frac{2k}{m-1}\Big)^2
=\frac{2km^2}{k-1}+4km+\frac{2mk^2}{m-1}.
\]
Divide by $n^2$ and use $k/n\to\gamma$, $m/n\to 1-\gamma$, $k/(k-1)\to1$, $m/(m-1)\to1$ uniformly on $[\eta,1-\eta]$ to obtain $2(1-\gamma)^2+4\gamma(1-\gamma)+2\gamma^2=2$. For \eqref{eq:row-l2}, if $i\le k$,
\[
\sum_{j\ne i} a_{n,k}(i,j)^2
=(k-1)\Big(\tfrac{2m}{k-1}\Big)^2 + m\cdot 2^2
=\frac{4m^2}{k-1}+4m
\le C n
\]
uniformly in $\gamma\in[\eta,1-\eta]$; the case $i>k$ is analogous.
\end{proof}

\begin{Lemma}\label{lem21}
Under $H_0$ and Assumption \ref{A:moment}--\ref{A:Ustat},
\begin{equation}
s_{\Psi,n}^2 \xrightarrow{p} \sigma_\Psi^2 := \mathbb{E}\big[\Psi(X,X')^2\big],
\end{equation}
where $\Psi$ is the degenerate kernel in the Hoeffding decomposition of $h$.
\end{Lemma}
\begin{proof}
Since $\hat{\Psi}_{ij}$ is the usual double-centered estimator of the canonical (degenerate) kernel
$\Psi(X_i,X_j)$, the statistic $s_{\Psi,n}^2$ is a (pooled) $V$-statistic of order $2$ with kernel
$(x,y)\mapsto \Psi(x,y)^2$. Square-integrability in Assumption~\ref{A:Ustat} yields the $V$-statistic LLN, giving the stated
convergence.
\end{proof}

\begin{theorem}\label{theorem1}
Suppose Assumptions \ref{A:sampling}–\ref{A:Ustat} hold. Then under $H_0$, uniformly for $k/n\in[\eta,1-\eta]$,
\[
\mathcal{Z}_{n,k}\ \xrightarrow{d}\ \mathcal{N}(0,1).
\]
\end{theorem}

\begin{proof}
We write the Hoeffding decomposition
\[
h(x,y)=\mu+\psi(x)+\psi(y)+\Psi(x,y),
\]
with $\Ex\psi(X)=0$, $\Ex\Psi(x,X)=\Ex\Psi(X,y)=0$ for all $x,y$, and put
\[
S_{n,k} \;:=\; \frac{k(n-k)}{n}\,E_k.
\]
By \eqref{eq:symm-id} and \eqref{eq:rowsums},  $S_{n,k}$ can be written as a weighted degenerate $U$-statistic,
\[
S_{n,k}
= \frac{1}{n}\sum_{i<j} a_{n,k}(i,j)\{\mu+\psi(X_i)+\psi(X_j)+\Psi(X_i,X_j)\}
= \frac{1}{n}\sum_{i<j} a_{n,k}(i,j)\,\Psi(X_i,X_j).
\]
Let 
\[
\sigma_\Psi^2:=\Var\{\Psi(X,X')\}, \quad \text{and} \quad
V_{n,k}^2 \;:=\; \frac{1}{n^2}\sum_{i<j} a_{n,k}(i,j)^2\,\sigma_\Psi^2.
\]
By Lemma~\ref{lem:weights}, $V_{n,k}^2\to 2\,\sigma_\Psi^2$ uniformly for $k/n\in[\eta,1-\eta]$.
Moreover, since
\[
\max_{1\le i\le n}\ \frac{1}{n^2}\sum_{j\ne i} a_{n,k}(i,j)^2 \;\le\; C/n
\quad\text{uniformly for~~}k/n\in[\eta,1-\eta]
\]
for some constant $C<\infty$, the Lindeberg–type condition for generalized quadratic forms holds. Hence,
\[
\frac{S_{n,k}}{V_{n,k}} \ \xrightarrow{d}\ \mathcal{N}(0,1)
\quad\text{uniformly for~~} k/n \in[\eta,1-\eta].
\]

By the Law of Large Numbers for $U$-statistics, the estimator defined in (\ref{dec26}), by Lemma \ref{lem21}, satisfies $s_{\Psi,n}^2 \xrightarrow{p} \sigma_{\Psi}^2$, and applying Lemma~\ref{lem:weights} yields
$\dfrac{1}{n^2}\sum_{i<j}a_{n,k}^2(i,j)\to 2$, hence $V_{n,k}/(\sqrt{2}\,s_{\Psi,n})\xrightarrow{p} 1$ uniformly.
Therefore
\[
\mathcal{Z}_{n,k}
=
\frac{S_{n,k}}{\sqrt{2}\,s_{\Psi,n}}
=
\Big(\frac{S_{n,k}}{V_{n,k}}\Big)\Big(\frac{V_{n,k}}{\sqrt{2}\,s_{\Psi,n}}\Big)
\Rightarrow N(0,1),
\]
uniformly over $k/n\in[\eta,1-\eta]$.
\end{proof}

\begin{Lemma}\label{lem:ED_mixture}
Let $F_1\neq F_2$ be distributions on $\mathbb{R}^p$ with
$\mathbb{E}\|X\|^\omega<\infty$ for $0<\omega<2$.
For $a,b\in[0,1]$, define $P_a:=aF_1+(1-a)F_2$ and $P_b:=bF_1+(1-b)F_2$.
Let $\Delta:=\mathcal{E}(F_1,F_2)$.
Then
\[
\mathcal{E}(P_a,P_b)=(a-b)^2\,\Delta.
\]
\end{Lemma}

\begin{proof}
\sloppy Using the characteristic-function representation of energy distance,
$$\mathcal{E}(P,Q)=c_{\omega,p}\int_{\mathbb{R}^p}
|\varphi_P(t)-\varphi_Q(t)|^2\,\|t\|^{-(p+\omega)}\,dt,$$
we have $\varphi_{P_a}-\varphi_{P_b}=(a-b)(\varphi_{F_1}-\varphi_{F_2})$.
Substituting yields the claim.
\end{proof}

\begin{theorem}\label{theorem2}
Assume $H_1$ and Assumptions~\ref{A:sampling}–\ref{A:delta}, with $\cp/n\to\tstar\in(\eta,1-\eta)$.
Let $\gamma:=k/n\in[\eta,1-\eta]$ and recall
\[
E_k = 2U_{XY}(k)-U_{XX}(k)-U_{YY}(k),
\qquad
S_{n,k}:=\frac{k(n-k)}{n}E_k,
\qquad
Z_{n,k}:=\frac{S_{n,k}}{\sqrt{2}\,s_{\Psi,n}}.
\]
Define $\mu_{11}:=\EE\norm{X-X'}$, $\mu_{22}:=\EE\norm{Y-Y'}$, $\mu_{12}:=\EE\norm{X-Y}$ and
$\Delta:=2\mu_{12}-\mu_{11}-\mu_{22}>0$.

\begin{enumerate}
\item[(a)] As $n\to\infty$, for any $\gamma\in(\eta,1-\eta)$,
\[
\EE\!\left[\frac{S_{n,\lfloor n\gamma\rfloor}}{n}\right]
=
\Delta\,G(\gamma,\tstar) + o(1),
\]
where
\[
G(\gamma,\tstar)
=
\begin{cases}
(1-\tstar)^2\,\dfrac{\gamma}{1-\gamma}, & \gamma\le \tstar,\\[1.2ex]
(\tstar)^2\,\dfrac{1-\gamma}{\gamma}, & \gamma\ge \tstar.
\end{cases}
\]

Moreover, $G(\gamma,\tstar)$ is strictly increasing on $[\eta,\tstar]$ and strictly decreasing on $[\tstar,1-\eta]$,
hence it is uniquely maximized at $\gamma=\tstar$, with $G(\tstar,\tstar)=\tstar(1-\tstar)$.

\item[(b)] The pooled scale estimator satisfies
\[
s_{\Psi,n}^2 \xrightarrow{p} \sigma_{\Psi,\tstar}^2 := \EE\!\left[\Psi_{\tstar}(W,W')^2\right]\in(0,\infty),
\]
where $W,W'\iid F_{\tstar}:=\tstar F_1+(1-\tstar)F_2$ and $\Psi_{\tstar}$ denotes the degenerate kernel defined
with respect to $F_{\tstar}$.

\item[(c)] At $k=\cp$,
\[
\frac{\mathcal{Z}_{n,\cp}}{n} \xrightarrow{p}
\frac{\tstar(1-\tstar)\Delta}{\sqrt{2}\,\sigma_{\Psi,\tstar}}
\quad\text{and in particular}\quad
\mathcal{Z}_{n,\cp} \xrightarrow{p} \infty.
\]

\item[(d)]  If $\hat{k}=\arg\max_{k\in\Keta}|\mathcal{Z}_{n,k}|$, then $\hat{k}/n \xrightarrow{p} \tstar$.
\end{enumerate}
\end{theorem}

\begin{proof}\hfill

\begin{enumerate}
\item[(a)] Fix $\gamma\in(\eta,1-\eta)$ and write $\cp/n\to\tstar$.
When splitting at $k=\lfloor n\gamma\rfloor$, the first segment and second segment converge in empirical
proportions to mixtures of $F_1$ and $F_2$.

If $\gamma\le\tstar$, segment~1 is asymptotically pure $F_1$, while segment~2 is the mixture
\[
F^{(2)}_{\gamma} = \frac{\tstar-\gamma}{1-\gamma}F_1 + \frac{1-\tstar}{1-\gamma}F_2.
\]
If $\gamma\ge\tstar$, segment~2 is asymptotically pure $F_2$, while segment~1 is the mixture
\[
F^{(1)}_{\gamma} = \frac{\tstar}{\gamma}F_1 + \frac{\gamma-\tstar}{\gamma}F_2.
\]
For mixtures of the form $pF_1+(1-p)F_2$, the population energy distance satisfies the identity
\[
\mathcal{E}\big(p_aF_1+(1-p_a)F_2,\; p_bF_1+(1-p_b)F_2\big) = (p_a-p_b)^2\,\Delta.
\]
Since $E_k$ is a (ratio-consistent) $U$-statistic estimator of the corresponding population energy distance,
$\EE[E_k]$ converges to the above mixture energy distance, yielding
\[
\EE[E_{\lfloor n\gamma\rfloor}]
=
\Delta\,H(\gamma,\tstar)+o(1),
\quad
H(\gamma,\tstar)=
\begin{cases}
\big(\dfrac{1-\tstar}{1-\gamma}\big)^2, & \gamma\le\tstar,\\[1.2ex]
\big(\dfrac{\tstar}{\gamma}\big)^2, & \gamma\ge\tstar.
\end{cases}
\]
Multiplying by $\frac{k(n-k)}{n}=n\gamma(1-\gamma)+o(n)$ yields
\[
\EE\!\left[\frac{S_{n,\lfloor n\gamma\rfloor}}{n}\right]
=
\gamma(1-\gamma)\Delta\,H(\gamma,\tstar)+o(1)
=
\Delta\,G(\gamma,\tstar)+o(1),
\]
with the stated piecewise expression for $G(\gamma,\tstar)$.
Monotonicity follows by differentiation:
for $\gamma<\tstar$, $G(\gamma,\tstar)=(1-\tstar)^2\,\gamma/(1-\gamma)$ is strictly increasing,
and for $\gamma>\tstar$, $G(\gamma,\tstar)=(\tstar)^2(1-\gamma)/\gamma$ is strictly decreasing.
Thus the unique maximizer is $\gamma=\tstar$ and $G(\tstar,\tstar)=\tstar(1-\tstar)$.

\item[(b)] Under $H_1$, the empirical distribution of $\{X_i\}_{i=1}^n$ converges to the mixture law $F_{\tstar}$,
and the double-centering defining $\hat{\Psi}_{ij}$ is a $V$-statistic centering that is consistent for the
$F_{\tstar}$-centered degenerate kernel $\Psi_{\tstar}$. Hence $s_{\Psi,n}^2$ is a $V$-statistic estimator of
$\EE[\Psi_{\tstar}(W,W')^2]$, implying $s_{\Psi,n}^2 \xrightarrow{p} \sigma_{\Psi,\tstar}^2$.

\item[(c)] By Part (a) at $\gamma=\tstar$, $\EE[S_{n,\cp}]/n\to \Delta\,\tstar(1-\tstar)$, while by (b)
$s_{\Psi,n} \xrightarrow{p} \sigma_{\Psi,\tstar}$. Therefore $Z_{n,\cp}/n \xrightarrow{p} \tstar(1-\tstar)\Delta/(\sqrt{2}\sigma_{\Psi,\tstar})$,
and $Z_{n,\cp} \xrightarrow{p} \infty$.

\item[(d)] By Part (a), the deterministic function $G(\gamma,\tstar)$ has a unique maximizer at $\gamma=\tstar$ and is strictly
separated away from $\tstar$ on $[\eta,1-\eta]$. Combining this with the scaling in Parts (b)--(c) yields a standard
argmax-consistency argument for scan statistics, giving $\hat{k}/n \xrightarrow{p} \tstar$.
\end{enumerate}
\end{proof}

\begin{Corollary}\label{cor:scan_consistency}
Under $H_1$ and Assumptions~\ref{A:sampling}–\ref{A:delta}, $T_n:=\max_{k\in K_\eta}|\mathcal{Z}_{n,k}|\to\infty$ in probability.
Consequently, the permutation-calibrated test that rejects for $T_n>c_{\alpha,n}$
is consistent, $\mathbb{P}_{H_1}(T_n>c_{\alpha,n})\to 1$.
\end{Corollary}

While Theorem~\ref{theorem1} establishes the asymptotic normality of $\mathcal{Z}_{n,k}$ for a fixed $k$,
change-point detection requires scanning over the interior candidate set $\mathcal{K}_\eta$ and aggregating
the evidence across dependent split-specific statistics. We therefore define the global scan statistic
\begin{equation*}
\begin{aligned}
T_n := \max_{k\in\mathcal{K}_\eta}\,|\mathcal{Z}_{n,k}|,
\end{aligned}
\end{equation*}
and the corresponding maximizer
\begin{equation*}
\begin{aligned}
\hat{k}:=\arg\max_{k\in\mathcal{K}_\eta}\,|\mathcal{Z}_{n,k}|.
\end{aligned}
\end{equation*}
To obtain critical values for $T_n$, we calibrate the test using a permutation procedure with $L$ permutations; throughout the simulation study, we take $L=999$. \\

\begin{remark}\label{rem:process_limit}
While Theorem~\ref{theorem1} establishes the asymptotic normality of $\mathcal{Z}_{n,k}$ for a fixed split point $k$,
the global test statistic
\[
T_n=\max_{k\in\mathcal{K}_\eta}\bigl|\mathcal{Z}_{n,k}\bigr|
\]
depends on the \emph{maximum} of a dependent sequence. Under $H_0$, a functional central limit theorem for
$U$-statistics (see, e.g., \cite{cs1997, DehlingVukWendler2022}) implies that the process
$\{\mathcal{Z}_{n,\lfloor nt\rfloor}: t\in[\eta,1-\eta]\}$ converges weakly to a standardized Brownian-bridge limit,
$\{B(t)/\sqrt{t(1-t)}: t\in[\eta,1-\eta]\}$. Consequently,
\begin{equation}
\begin{aligned}
T_n \xrightarrow{d} \sup_{t\in[\eta,1-\eta]} \frac{|B(t)|}{\sqrt{t(1-t)}}.
\end{aligned}
\end{equation}
Although critical values for this supremum can be approximated using extreme-value theory for Gaussian processes
(e.g., \cite{gombay1995}), such approximations may converge slowly in finite samples. We therefore calibrate $T_n$
using a permutation procedure, which provides accurate finite-sample size control and remains asymptotically valid
under $H_0$ by exchangeability.
\end{remark}

\begin{Lemma}\label{lem:perm_level}
Under $H_0$, if $X_1,\dots,X_n$ are i.i.d.\ then, for a uniform random permutation $\pi$,
\[
(X_1,\dots,X_n)\stackrel{d}{=}(X_{\pi(1)},\dots,X_{\pi(n)}).
\]
Consequently, the permutation critical value $c_{\alpha,n}$ defined above yields a level-$\alpha$ test.
\end{Lemma}

\begin{proof}
Under $H_0$, the joint distribution of $(X_1,\dots,X_n)$ is exchangeable, so
$(X_1,\dots,X_n)\stackrel{d}{=}(X_{\pi(1)},\dots,X_{\pi(n)})$. Hence, conditional on the observed sample,
the permutation distribution of $T_n(X_{\pi(1)},\dots,X_{\pi(n)})$ matches the null distribution of $T_n$,
up to discreteness from using finitely many permutations. Therefore,
$\mathbb{P}_{H_0}(T_n>c_{\alpha,n})\le \alpha$.
\end{proof}

We reject $H_0$ at level $\alpha$ if $T_n>c_{\alpha,n}$, where $c_{\alpha,n}$ is obtained by resampling under $H_0$.
Specifically, for $l=1,\dots,L$, we draw an independent uniform random permutation $\pi_l$ of $\{1,\dots,n\}$,
recompute the scan statistic $T_n^{(l)}$ from the permuted sample $(X_{\pi_l(1)},\dots,X_{\pi_l(n)})$, and set
$c_{\alpha,n}$ to be the empirical $(1-\alpha)$ quantile of $\{T_n^{(l)}\}_{l=1}^L$. We then reject $H_0$ if
the observed $T_n$ exceeds $c_{\alpha,n}$. The complete testing procedure for single change-point detection is summarized in Algorithm~\ref{alg:ED_single}.

\begin{algorithm}[H]
\caption{ED scan statistic with permutation calibration}
\label{alg:ED_single}
\begin{algorithmic}[1]
    \Require Data $X_1,\dots,X_n$, trimming $\eta\in(0,1/2)$, level $\alpha\in(0,1)$, permutations $L$.
    
    \State Compute $s_{\Psi,n}$ from the pooled double-centered kernel $\hat{\Psi}_{ij}$.
    
    \For{$k \in \mathcal{K}_\eta=\{\lceil \eta n\rceil,\dots,\lfloor(1-\eta)n\rfloor\}$}
        \State Compute $E_k=2U_{XY}(k)-U_{XX}(k)-U_{YY}(k)$
        \State Compute statistic:
        \[
            \mathcal{Z}_{n,k}=\frac{k(n-k)}{\sqrt{2}\,n\,s_{\Psi,n}}\,E_k
        \]
    \EndFor
    
    \State Set $\hat{k}=\arg\max_{k\in\mathcal{K}_\eta}|\mathcal{Z}_{n,k}|$ and $T_n=\max_{k\in\mathcal{K}_\eta}|\mathcal{Z}_{n,k}|$.
    
    \For{$l=1,\dots,L$}
        \State Draw a uniform random permutation $\pi_l$
        \State Recompute $T_n^{(l)}$ from $(X_{\pi_l(1)},\dots,X_{\pi_l(n)})$.
    \EndFor
    
    \State Let $c_{\alpha,n}$ be the empirical $(1-\alpha)$ quantile of $\{T_n^{(l)}\}_{l=1}^L$.
    \State Reject $H_0$ if $T_n>c_{\alpha,n}$.
    
    \Ensure Decision (reject/do not reject), and estimator $\hat{k}$ when rejecting.
\end{algorithmic}
\end{algorithm}

\begin{remark}
When Algorithm~\ref{alg:ED_single} is applied recursively (binary segmentation),
testing at level $\alpha$ at every node does not in general control the overall
family-wise error rate.
A simple conservative option is to test each segment at level
$\alpha_{\mathrm{seg}}=\alpha/M$, where $M$ is an upper bound on the number of
segment tests performed (Bonferroni).
Alternatively, one may use multiscale procedures (e.g., SMUCE) or WBS-style
schemes with global thresholds.
\end{remark}

Although the theory in Section~2 focuses on a single change point, to detect multiple changes, we apply Algorithm~\ref{alg:ED_single} recursively within segments (binary segmentation).
For a segment $[s,e]$ with length $n_{s,e}=e-s+1$, define $\mathcal{K}_\eta(s,e)$ analogously and compute
$T_{s,e}=\max_{k\in\mathcal{K}_\eta(s,e)}|\mathcal{Z}_{s,e}(k)|$ using only $\{X_s,\dots,X_e\}$.
If $T_{s,e}>c_{\alpha,n_{s,e}}$ and $n_{s,e}\ge n_{\min}$, we split at $\hat{k}_{s,e}$ and recurse on the left and right
subsegments; otherwise, we stop. In all experiments we set $n_{\min}=\lceil 2\eta n\rceil$ to avoid degenerate splits.

\section{Simulation Study}
In this section, we conduct a simulation study to assess the performance of the proposed change point detection method under various scenarios to understand its sensitivity and robustness. We consider sample sizes $n \in \{20, 30, 50, 100, 200\}$. For each sample size, potential change points are introduced at locations representing proportions of the data including at $k \in \{0.1n,\, 0.2n,\, 0.3n,\, 0.4n,\, 0.5n\}$, corresponding to $ 10\% $, $ 20\% $, $ 30\% $, $ 40\% $, and $ 50\% $ of the total sample size. Further, data are generated from three different distributions to explore the method's performance under different distributional assumptions, including the standard normal distribution $N(0,1)$, the standard skew-normal distribution $\mathrm{SN}(0,1,1)$, and the standard exponential distribution $\mathrm{Exp}(1)$. For the pre-change data, we use these standard distributions.  In the Type I error simulations, we generate the entire dataset from the pre-change distribution without any change point to estimate the false positive rate of the method under the null hypothesis. \\ %The significance level for the hypothesis tests is set at $ \alpha = 0.05$.\\

For power analysis, the post-change data is generated under the following settings. To evaluate the method's ability to detect changes in the mean, we consider mean shift magnitudes of $ \mu_{\text{post}} \in \{0.5,\, 1,\, 1.5,\, 2\}$. Similarly, to assess sensitivity to changes in variance, variance shift of $\sigma^2_{\text{post}} \in \{1.5,\, 2,\, 2.5,\, 3\}$ are used. Since the Skew Normal distribution is a 3-parameter distribution, in addition to mean and variance, we consider the shape parameter $\kappa_{\text{post}} \in \{1.5,\, 2,\, 2.5,\, 3\}$.  We also considered the simultaneous changes. For the normal distribution,  the mean and variance, for the Skew Normal distribution, we considered the mean, variance, and shape parameters simultaneously.  We evaluate the proposed method against existing non-parametric approaches in the literature, including E-Divisive (ECP), PELT, and Fused Lasso (FL) methods. For the Fused Lasso (FL) implementation, we utilized the \texttt{genlasso} R package. The penalty parameter was selected using $10$-fold cross-validation to minimize the mean squared error. It is worth noting that cross-validation tends to select a model with high predictive accuracy but may include spurious change points, leading to the inflated Type I error rates observed in Table~1. Similarly, for the E-Divisive (ECP) method, we utilized the standard settings provided in the \texttt{ecp} R package with the significance level $\alpha = 0.05$ for the permutation test.
\\

For the power simulations, we introduce a change point at location $ k $ in each dataset. The data prior to the change point ($ 1,\ldots,k $) are generated from the pre-change distribution. The data following the change point ($ k+1,\ldots,n $) are generated from the same distribution but with the specified changes in mean or variance. Specifically, for mean change scenarios, the post-change observations are generated from a distribution with a mean shifted by $ \mu_{\text{post}} $, and for variance change scenarios, the post-change observations have their variance from $ \sigma^2_{\text{post}} $. The simulation results are based on 1000 iterations. In all simulation experiments, we set $\alpha=0.05$ and $\eta \in (0,1/2)$ to restrict candidate split points to $\mathcal{K}_\eta$. For ED, the test rejects $H_0$ when $T_n=\max_{k\in \mathcal{K}_\eta}|\mathcal{Z}_{n,k}|$ exceeds its permutation-based critical value $c_{\alpha,n}$ computed from $L=999$ random permutations under $H_0$. This calibration is the basis for the empirical Type I error reported in Table~\ref{ty1}.\\

{\setlength{\tabcolsep}{1em}
\begin{table}[H]
\footnotesize
\caption{Empirical size comparison with nominal level $\alpha = 0.05$ for various distributions (normal, skew normal, exponential) with sample sizes $n \in \{20, 30, 50, 100, 200\}$.}
\centering
\begin{tabular}{ccccccccccccc}
\hline
 & & \multicolumn{4}{c}{Methods} \\ \hline
Distribution & \multirow{1}{*}{$n$} & ED & FL & ECP & PELT \\ \hline
$\text{N}(0,1)$ 
 & 20  & 0.048 & 0.040 & 0.054 & 0.056 \\
 & 30  & 0.056 & 0.057 & 0.048 & 0.055 \\
 & 50  & 0.060 & 0.047 & 0.048 & 0.043 \\
 & 100 & 0.043 & 0.096 & 0.048 & 0.041 \\
 & 200 & 0.039 & 0.039 & 0.035 & 0.038 \\ \hline
$\text{SN}(0,1,1)$ 
& 20  & 0.047 & 0.073 & 0.054 & 0.056 \\
 & 30  & 0.034 & 0.082 & 0.060 & 0.036 \\
 & 50  & 0.040 & 0.031 & 0.053 & 0.052 \\
 & 100 & 0.061 & 0.094 & 0.046 & 0.050 \\
 & 200 & 0.055 & 0.065 & 0.053 & 0.040 \\\hline
$\text{Exp}(1)$ 
 & 20  & 0.057 & 0.053 & 0.062 & 0.061 \\
 & 30  & 0.042 & 0.036 & 0.050 & 0.056 \\
 & 50  & 0.047 & 0.064 & 0.054 & 0.029 \\
 & 100 & 0.051 & 0.041 & 0.058 & 0.039 \\
 & 200 & 0.042 & 0.057 & 0.038 & 0.049 \\ \hline
\end{tabular}
\label{ty1}
\end{table}}

Table~\ref{ty1} reports empirical Type~I error at the nominal level $\alpha=0.05$ for the proposed ED and three competitors across $n\in\{20,30,50,100,200\}$ under three baseline distributions. Overall, ED exhibits accurate and stable size control across all scenarios. The ED Type~I error rate lies between $0.039$ and $0.060$, which is within about $\pm 0.01$ of the nominal level. This stability makes ED competitive with the best existing methods and, in several regimes, superior in accuracy. For the normal distribution $N(0,1)$, ED stays close to $0.05$ for all sample sizes, ranging from $0.039$ to $0.060$. The FL method shows a notable overshoot at $n=100$ with a value of $0.096$. The ECP method and the PELT method are generally near nominal but tend to be mildly conservative at larger sample sizes. Hence, under symmetry, ED provides well–controlled size and avoids the inflation observed for the FL method. \\

For the skew normal distribution $SN(0,1,1)$, ED remains well controlled across sample sizes, being closest to nominal at $n=20$ and $n=50$, mildly conservative at $n=30$ with a value of $0.034$, and mildly liberal at $n=100$ and $n=200$ with values of $0.061$ and $0.055$. The performance of the alternative methods varies more noticeably across settings. For instance, the FL method is liberal for several sample sizes (for example, $0.073$, $0.094$, and $0.065$), while the ECP method and the PELT method occasionally undershoot. ED thus provides the most consistent size across the skewed setting and remains at least competitive with the ECP method and the PELT method, while outperforming the FL method in terms of overall accuracy. For the exponential distribution $\mathrm{Exp}(1)$, which represents a heavily skewed case, ED remains close to the nominal level, for example, $0.057$ at $n=20$, $0.051$ at $n=100$, and $0.042$ at $n=200$. It typically falls between the conservative behavior of the PELT method at larger sample sizes and the occasional mild liberal tendencies of the ECP method and the FL method. Across distributions and sample sizes, ED achieves reliable Type~I error control with low variability and avoids the substantial deviations exhibited by the FL method.

\begin{figure}[H]
  \centering
  \includegraphics[width=0.8\textwidth]{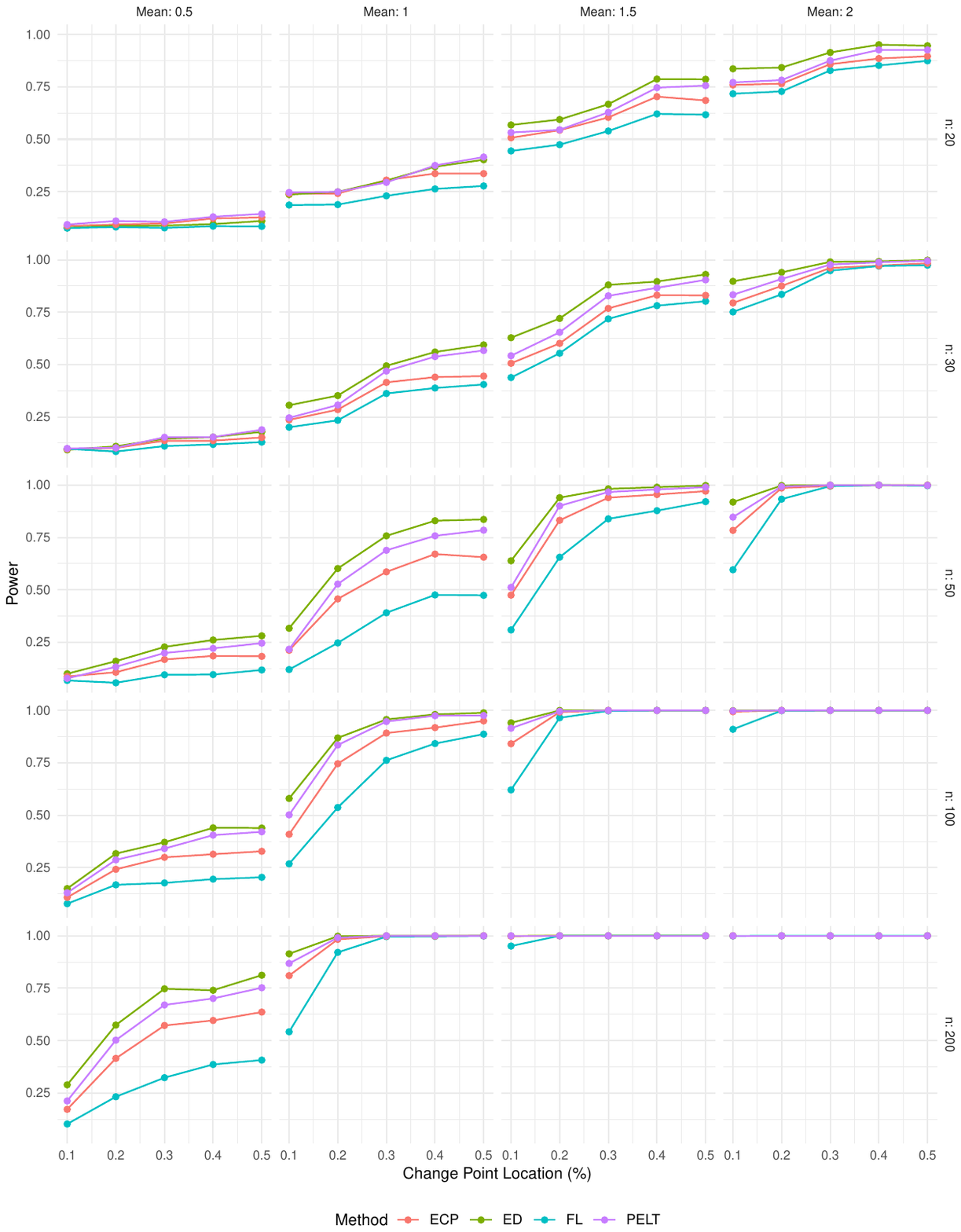}
  \caption{Power comparison for the normal distribution under a mean change $\mu_{\text{post}} \in \{0.5, 1, 1.5, 2\}$ across sample sizes $n \in \{20, 30, 50, 100, 200\}$ and change-point locations $k \in \{0.1n, 0.2n, 0.3n, 0.4n, 0.5n\}$.}
  \label{fig1}
\end{figure}

As shown in Figure~\ref{fig1}, under a normal distribution setting with a single change in the mean, the proposed ED procedure consistently competes with and often outperforms previously proposed methods. Power increases with the magnitude of the mean shift, with the sample size, and as the change point moves from the boundary toward the center, which is the expected pattern for all methods. For small shifts, ED already shows an advantage once the sample size is moderate. For example, with shift 0.5 and $n=50$, ED attains power 0.10 compared with 0.07 for FL, 0.09 for ECP, and 0.08 for PELT; at $n=100$ the gap widens, with ED at 0.15 while the next best method is 0.13. At $n=200$, ED continues to lead with power 0.29 versus 0.21 for PELT, 0.17 for ECP, and 0.10 for FL.  As the mean shift increases to 1, 1.5, and 2, ED remains highly competitive and is frequently the best performer across locations. With shift 1 and $n=50$, ED reaches 0.32, exceeding ECP at 0.21, PELT at 0.22, and FL at 0.12; by $n=100$, ED is at 0.58 while competitors range from 0.27 to 0.50. With shift 1.5 and $n=50$, ED attains 0.76 compared with 0.39 for FL, 0.59 for ECP, and 0.69 for PELT; at $n=100$, the power of ED is close to one across most locations. For the largest shift of 2, all methods approach unit power at moderate and large sample sizes, and ED matches or slightly exceeds the best-performing alternative.

\begin{figure}[H]
  \centering
  \includegraphics[width=0.8\textwidth]{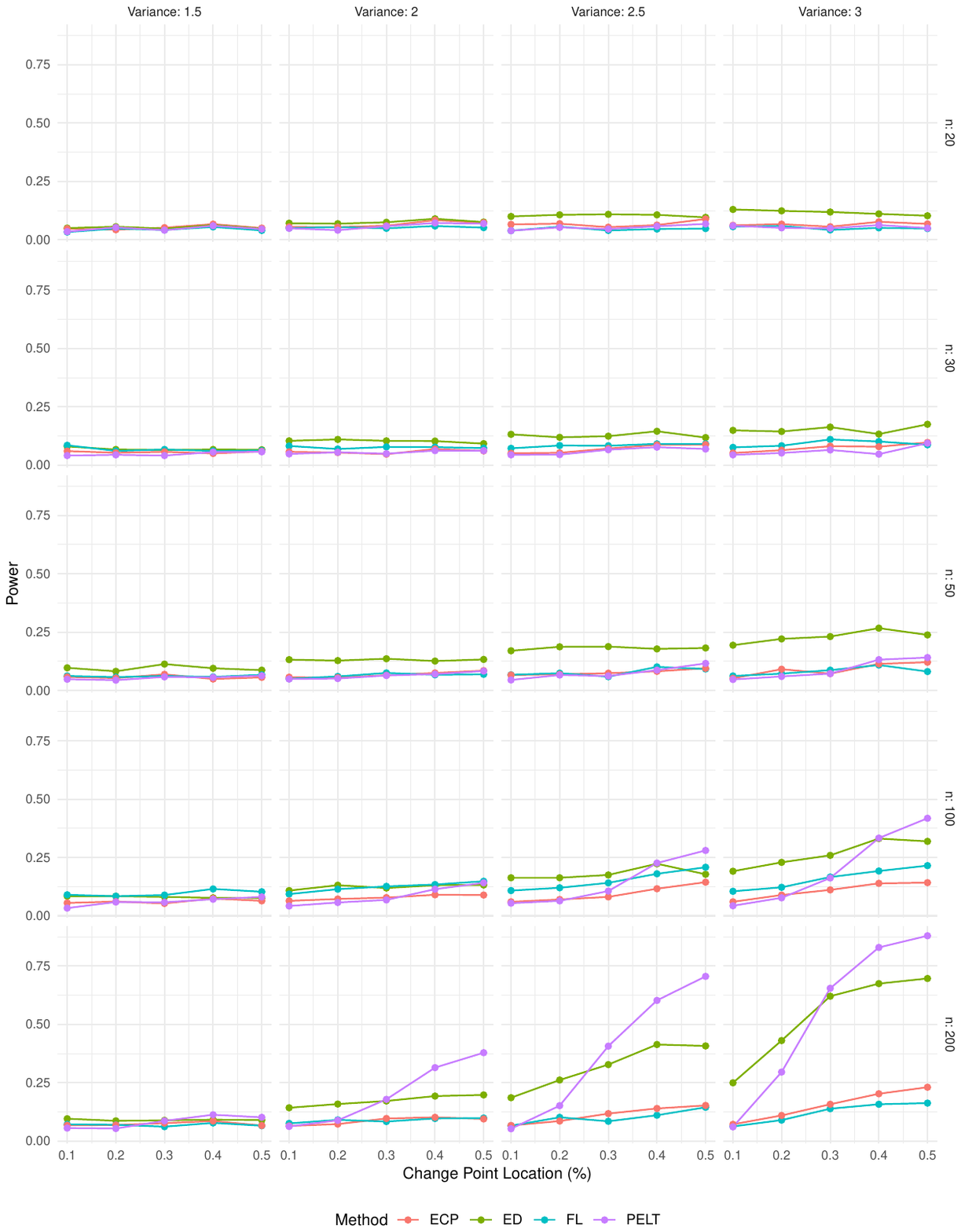}
  \caption{Power comparison for the normal distribution under a variance change $\sigma^2_{\text{post}} \in \{1.5, 2, 2.5, 3\}$ across sample sizes $n \in \{20, 30, 50, 100, 200\}$ and change-point locations $k \in \{0.1n, 0.2n, 0.3n, 0.4n, 0.5n\}$.}
  \label{fig2}
\end{figure}

Figure~\ref{fig2} compares power across methods for a normal distribution setting with a single change in variance and shows that the proposed ED procedure exhibits consistently strong performance relative to competing approaches. Power increases with the magnitude of the variance change, with the sample size, and as the change point moves from the boundary toward the center of the sequence. For small samples, such as $n=20$ and $n=30$, all procedures have limited power under modest variance changes, yet ED is typically the best or tied with the best across locations. At $n=50$, ED begins to separate more clearly for moderate variance contrasts, exceeding FL and ECP and matching or surpassing PELT at many interior locations. The boundary regions remain the most difficult, and ED maintains a visible advantage there as well. These patterns indicate that ED is competitive across the board and often outperforms the existing methods when detection is challenging. For larger samples, such as $n=100$ and $n=200$, power rises markedly for all methods, and ED remains a leading performer across change point locations. Under moderate variance shifts, ED attains higher power than FL and ECP and competes closely with PELT, frequently surpassing it away from the boundaries. When the variance jump is large, all methods approach unit power, and ED remains at least as effective as the strongest competitor throughout the range of locations. The performance of ED improves gradually and consistently with increasing sample size, a property that is advantageous for real-data applications. \\

Figure~\ref{fig3} demonstrates that when both the mean and the variance change under a normal distribution, the proposed ED performs competitively and frequently achieves the highest power across sample sizes and change-point locations. For moderate magnitudes such as (mean, variance) equal to (0.5, 1.5) or (1, 1.5), power is low at small samples, although ED is typically the top method or comparable to the best competitor. Power increases as the change moves from the boundary toward the center, and ED shows a consistent advantage in these boundary regions where detection is most difficult. As the sample size grows from \(n=20\) and \(n=30\) to \(n=50\), ED begins to separate more clearly, often exceeding FL and ECP and matching or outperforming PELT at interior locations. These patterns indicate that ED is better controlled to subtle joint shifts and that it maintains superior detection near the edges of the sequence.  For larger magnitudes such as (1.5, 2) and (2, 2.5), power rises rapidly with sample size, and ED remains a leading performer across locations. At \(n=100\), ED typically attains higher power than FL and ECP and competes closely with PELT, with visible gains away from the boundaries. At \(n=200\), all procedures approach unit power, and ED continues to match or slightly exceed the strongest competitor across most locations.

\begin{figure}[H]
  \centering
  \includegraphics[width=0.8\textwidth]{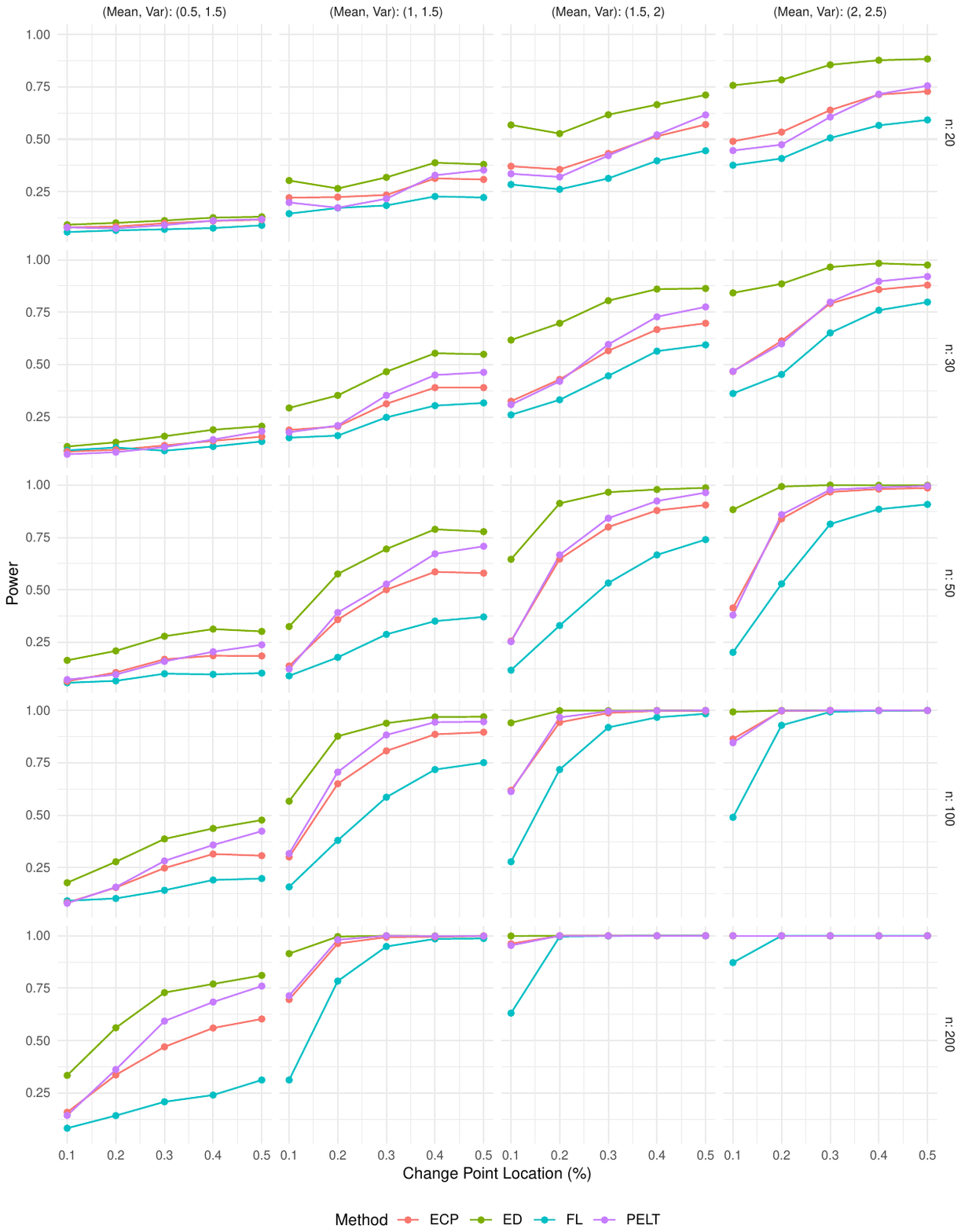}
  \caption{Power comparison for the normal distribution under simultaneous changes in mean and variance, with parameters $\mu_{\text{post}} \in \{0.5, 1, 1.5, 2\}$ and $\sigma^2_{\text{post}} \in \{1.5, 2, 2.5, 3\}$ across sample sizes $n \in \{20, 30, 50, 100, 200\}$ and change-point locations $k \in \{0.1n, 0.2n, 0.3n, 0.4n, 0.5n\}$}
  \label{fig3}
\end{figure}

Figure~\ref{fig4} illustrates that for a skew normal distribution with a single change in the mean, the proposed ED exhibits strong power performance across sample sizes, change point locations, and shift magnitudes. Power increases as the mean shift grows from 0.5 to 2, as the sample size increases from $n=20$ to $n=200$, and as the change point moves from the boundary toward the center of the sequence. In small samples, all procedures have limited power for the smallest shifts, however, ED is typically at the top or tied with the best competitor across locations. At $n=50$, ED begins to separate more clearly for moderate shifts, exceeding FL and ECP and matching or outperforming PELT at many interior locations. The boundary regions remain the most difficult, and ED maintains a visible advantage there as well. These patterns indicate that ED is competitive across the board and often outperforms the existing methods when detection is most challenging. For larger samples such as $n=100$ and $n=200$, power rises sharply for all methods, and ED remains a leading performer across change point locations. Under moderate shifts, ED attains higher power than FL and ECP and competes closely with PELT, frequently matching or surpassing its performance in boundaries and central portions of the sequence. When the mean shift is large, all procedures approach unit power, and ED remains at least as effective as the strongest competitor throughout the range of locations.

\begin{figure}[H]
  \centering
  \includegraphics[width=0.8\textwidth]{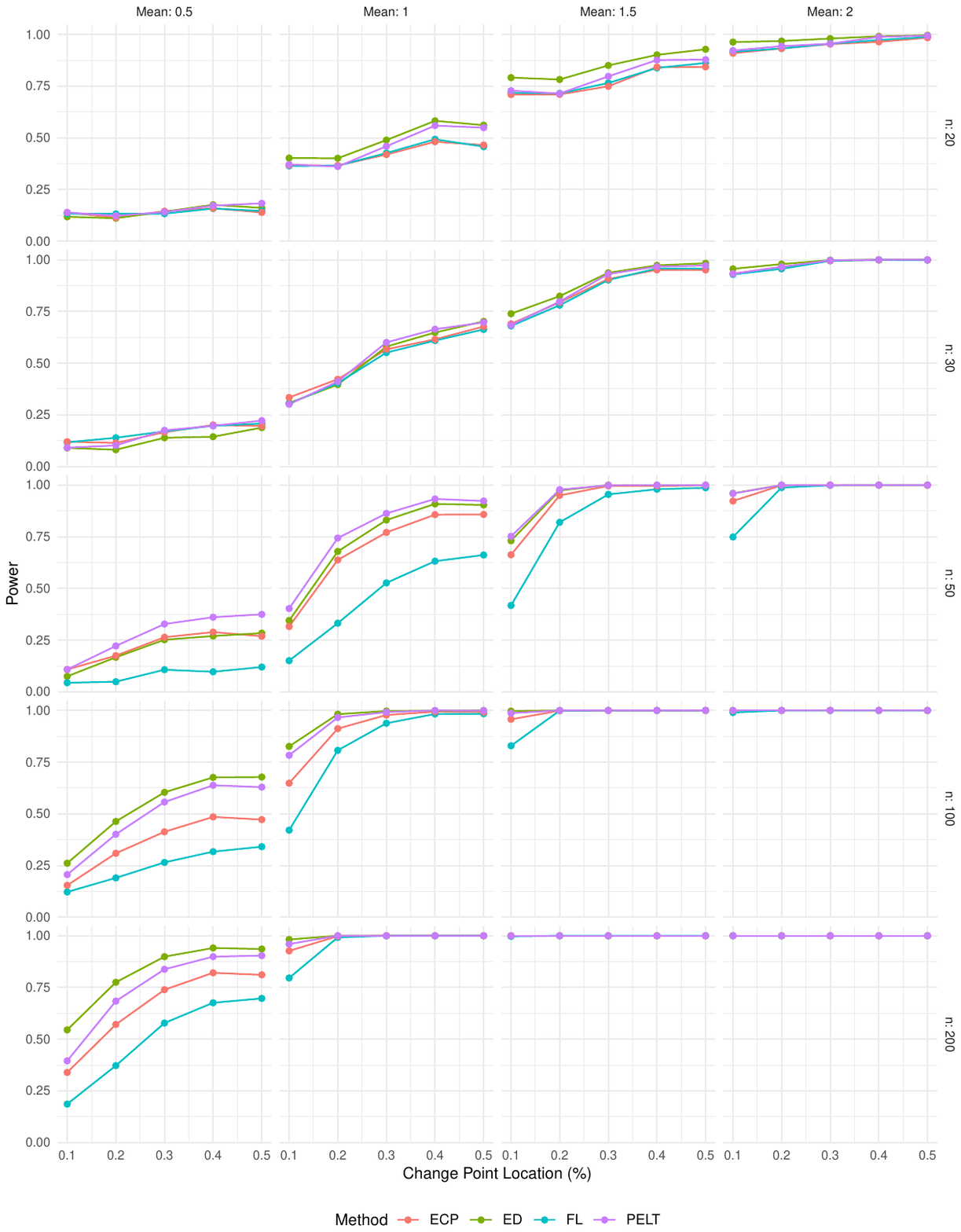}
  \caption{Power comparison for the skew normal distribution under a mean change $\mu_{\text{post}} \in \{0.5, 1, 1.5, 2\}$ across sample sizes $n \in \{20, 30, 50, 100, 200\}$ and change-point locations $k \in \{0.1n, 0.2n, 0.3n, 0.4n, 0.5n\}$ }
  \label{fig4}
\end{figure}

As shown in Figure~\ref{fig5}, for a skew normal distribution with a single change in variance, the proposed ED consistently provides strong power across sample sizes, locations, and magnitudes. Power increases as the variance ratio grows from 1.5 to 3, as the sample size increases from $n=20$ to $n=200$, and as the change point moves from the boundary toward the center of the sequence. Under the most challenging scenarios involving small samples and boundary change points, all methods exhibit limited power; however, ED is typically the leading procedure or comparable in performance to the top-performing method. At $n=50$, ED begins to separate more clearly for moderate variance shifts, outperforming FL and ECP and often matching or exceeding PELT within interior segments. This advantage remains evident across locations, which indicates that ED is better optimized to detect variance changes under skewness. Overall, ED competes strongly in easy regimes and frequently outperforms the existing methods when detection is most challenging. For larger samples such as $n=100$ and $n=200$, Figure~\ref{fig5} shows a marked rise in power for all procedures, with ED remaining a leading performer across the full range of change point locations. Under moderate variance contrasts, ED achieves higher power than FL and ECP and competes closely with PELT, frequently outperforming it at interior locations. When the variance jump is large, all methods approach unit power, and ED continues to be at least as effective as the strongest competitor throughout the domain of locations.

\begin{figure}[H]
  \centering
  \includegraphics[width=0.8\textwidth]{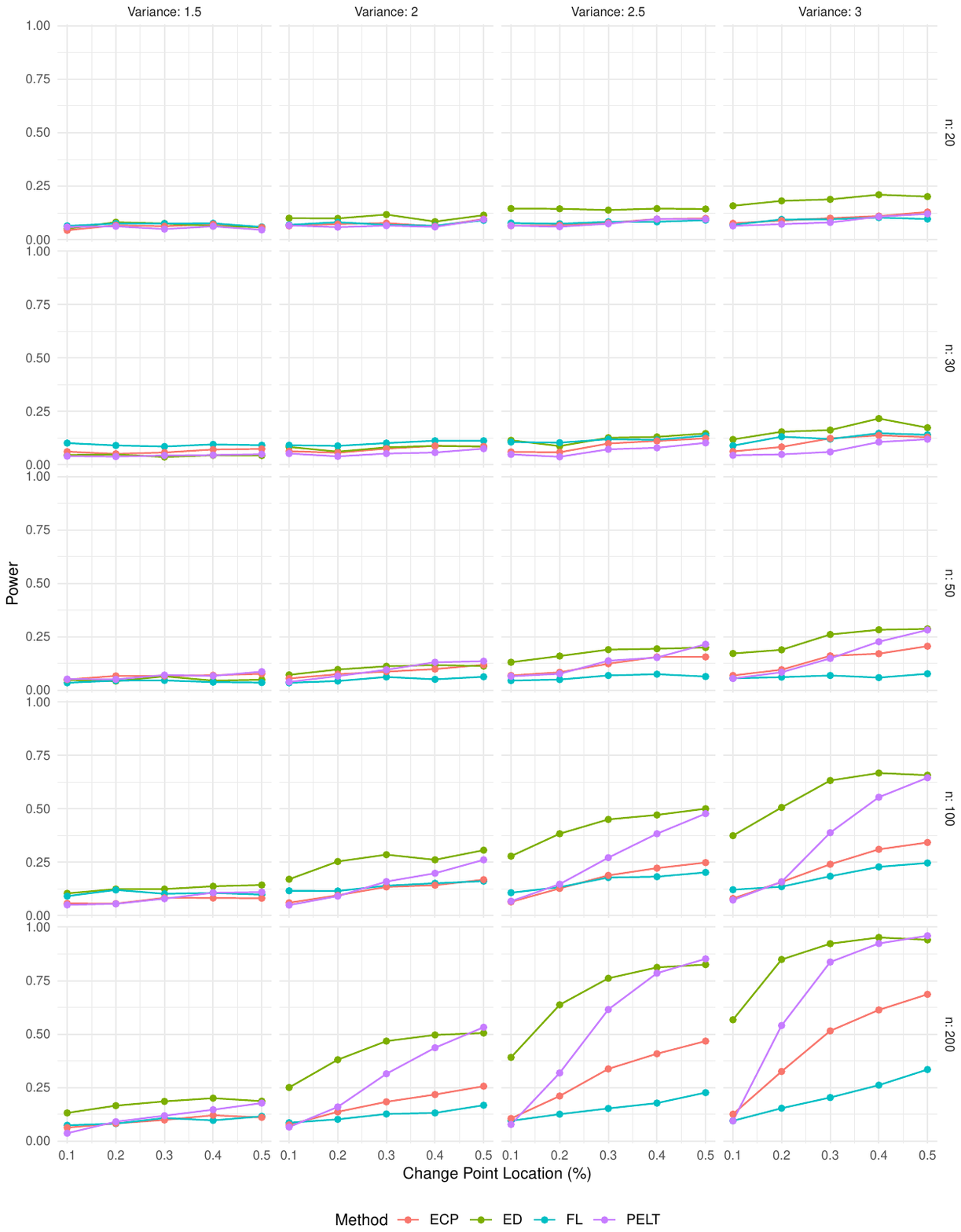}
  \caption{Power comparison for the skew normal distribution under a variance change $\sigma^2_{\text{post}} \in \{1.5, 2, 2.5, 3\}$ across sample sizes $n \in \{20, 30, 50, 100, 200\}$ and change-point locations $k \in \{0.1n, 0.2n, 0.3n, 0.4n, 0.5n\}$}
  \label{fig5}
\end{figure}

\begin{figure}[H]
  \centering
  \includegraphics[width=0.8\textwidth]{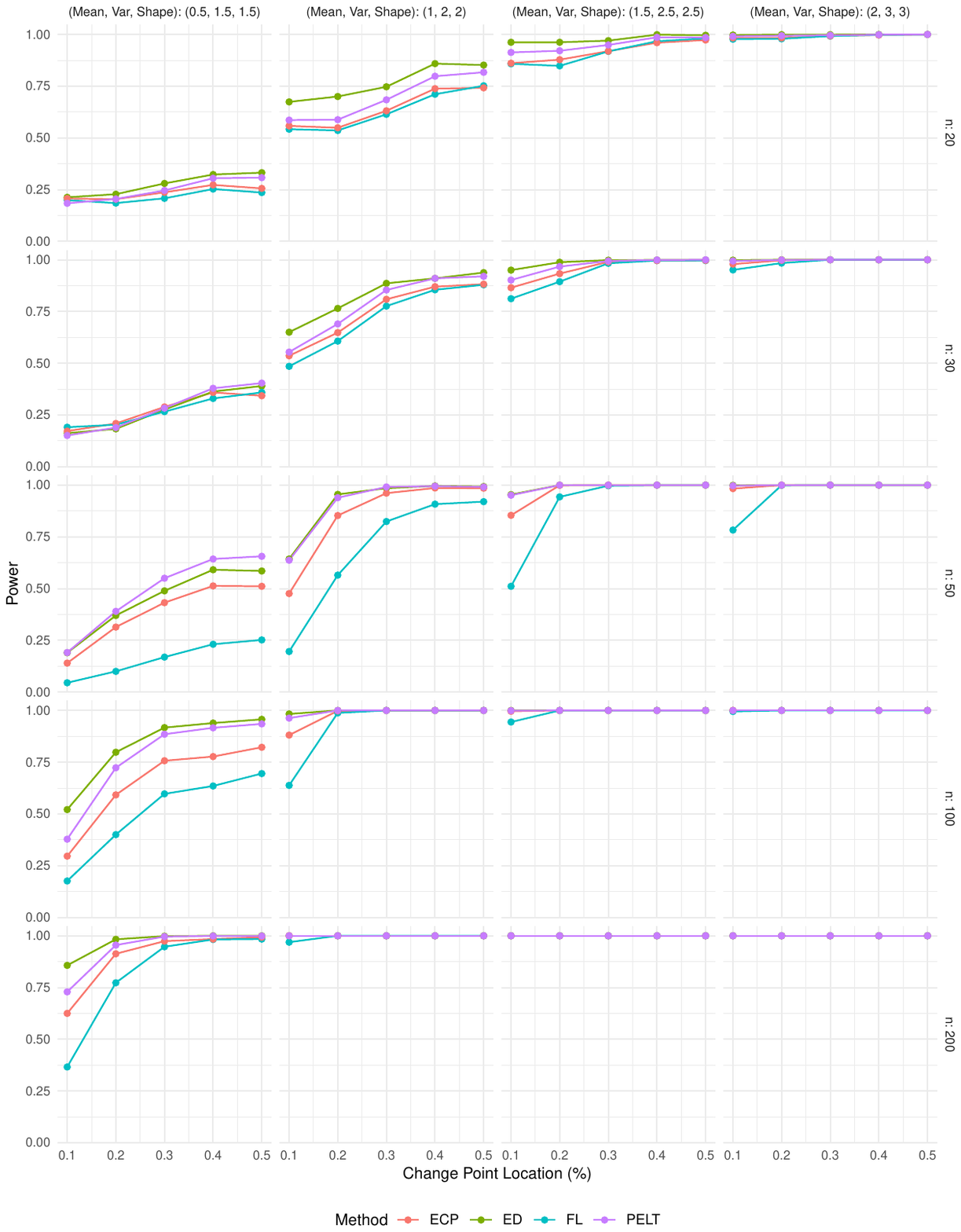}
  \caption{Power comparison for the skew normal distribution under simultaneous changes in mean, variance, and shape parameters, with $\mu_{\text{post}} \in \{0.5, 1, 1.5, 2\}$, $\sigma^2_{\text{post}} \in \{1.5, 2, 2.5, 3\}$, and $\kappa_{\text{post}} \in \{1.5, 2, 2.5, 3\}$ across sample sizes $n \in \{20, 30, 50, 100, 200\}$ and change-point locations $k \in \{0.1n, 0.2n, 0.3n, 0.4n, 0.5n\}$.}
  \label{fig6}
\end{figure}

Figure~\ref{fig6} shows that power rises monotonically as the change point moves from the boundary toward the center, for example, from $k_{\text{prop}}=0.1$ to $0.5$, with the steepest gains between $0.2$ and $0.4$. For the smallest parameter shift $(\Delta\mu,\Delta\sigma^2,\Delta\kappa)=(0.5,1.5,1.5)$, the proposed ED procedure already competes closely at $n\in\{20,30,50\}$ and begins to clearly outperform alternatives by $n=100$ across most locations (e.g., ED $=0.521$ vs.\ PELT $=0.378$ at $k_{\text{prop}}=0.1$), with dominance widening further at $n=200$. At these weaker effects, ED exhibits the largest marginal improvement as $k_{\text{prop}}$ increases, which demonstrates better sensitivity to interior changes where data from both segments contribute more evenly. Even when competitors briefly match ED at isolated settings, for example, $n=50$, $k_{\text{prop}}=0.1$, ED is never worse in a systematic way and typically becomes better as either the location moves inward or the sample size grows. Overall, for small-to-moderate shifts, ED provides the strongest combination of power growth in $k_{\text{prop}}$ and robust performance across increasing sample sizes, thus outperforming or at least tightly competing with ECP, FL, and PELT. \\

As the shift magnitude increases to $(1,2,2)$ and $(1.5,2.5,2.5)$, power curves for all methods climb sharply, whereas ED reaches high power earlier in $n$ and across all change-point locations, reflecting better detection at modest sample sizes. For example, at magnitude $(1,2,2)$ with $n=50$ and $k_{\text{prop}}=0.1$, ED attains $0.643$ power compared with FL at $0.196$ and PELT at $0.637$, demonstrating ED’s advantage when changes are near the boundary and data are limited. By $n=100$, ED is already near or at unit power for $(1,2,2)$ across locations, and at $(1.5,2.5,2.5)$ the ED curve reaches a plateau in performance, maintaining superiority over competing methods across all settings. At the largest magnitude $(2,3,3)$, all procedures are effectively perfect, but ED never lags and often achieves the plateau first. Taken together, these results demonstrate that ED is not only competitive but frequently better than ECP, FL, and PELT, achieving higher power sooner in $n$, maintaining strong performance for interior and boundary change points, and scaling favorably as the magnitude of the joint parameter shift increases.

\begin{figure}[H]
  \centering
  \includegraphics[width=0.8\textwidth]{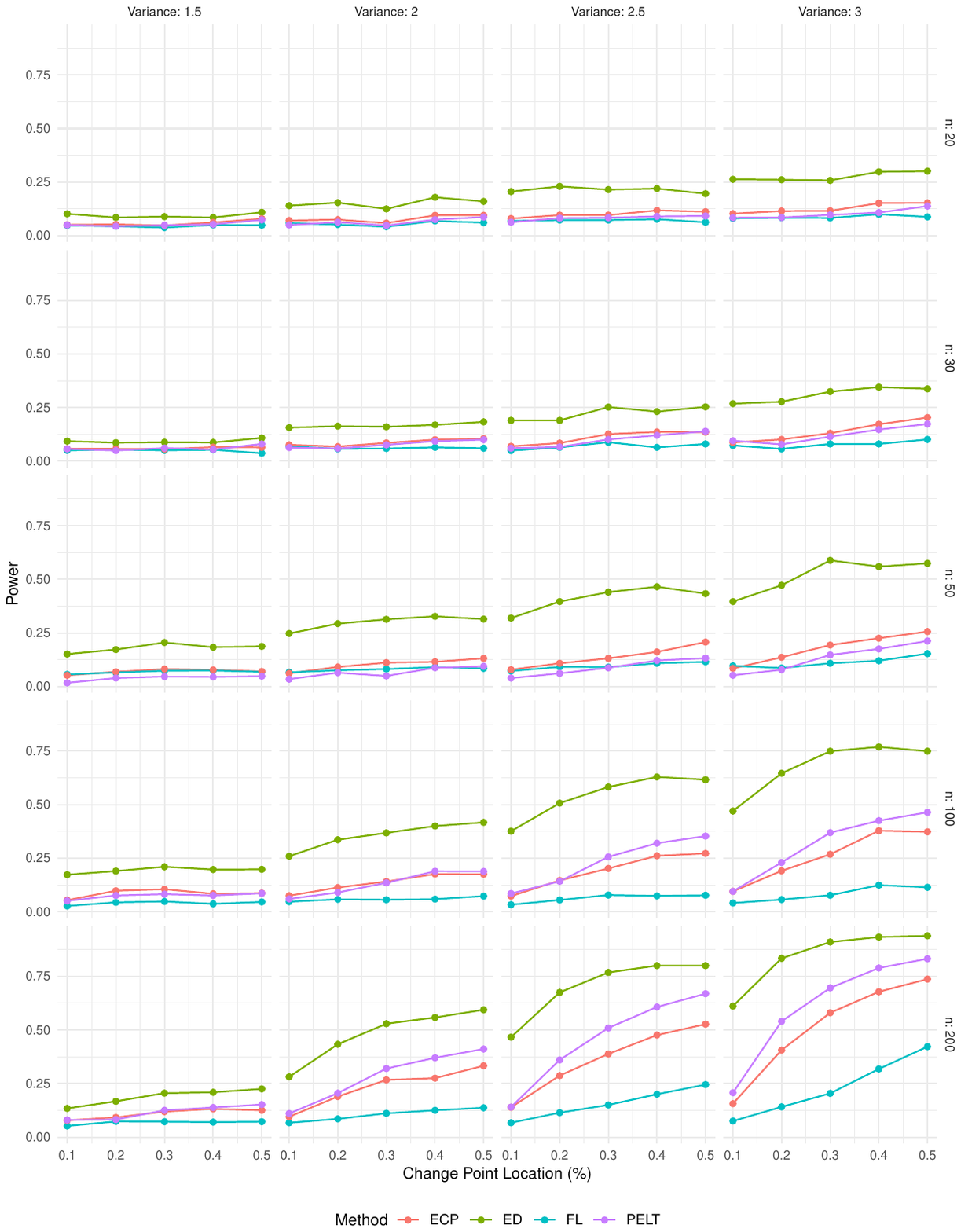}
  \caption{Power comparison for the exponential distribution under a variance parameter change $\sigma^2_{\text{post}} \in \{1.5, 2, 2.5, 3\}$ across sample sizes $n \in \{20, 30, 50, 100, 200\}$ and change-point locations $k \in \{0.1n, 0.2n, 0.3n, 0.4n, 0.5n\}$.}
  \label{fig7}
\end{figure}

Figure~\ref{fig7} presents the pattern of power comparison across sample size ($n$), change-point location ($k_{\text{prop}}$), and variance changes in the exponential distribution. It can be clearly seen that power rises as $n$ increases from 20 to 200 and also increases as the change occurs later in the sequence, for example, from $k_{\text{prop}}{=}0.1$ to $0.5$, and reflects the benefit of having more post-change observations to estimate the new scale. The proposed ED procedure consistently outperforms its competitors under modest variance changes (e.g., $\sigma^2_{\text{post}}=1.5$), where alternative methods, including ECP, FL, and PELT, often remain near nominal levels, while ED achieves higher power. For instance, at $n{=}100$, ED attains approximately $0.21$ at $k_{\text{prop}}{=}0.3$, compared with less than $0.11$ for FL, ECP, or PELT. As the variance shift strengthens (e.g., $\sigma^2_{\text{post}}=2$ and $2.5$), ED’s advantage becomes more pronounced and exhibits markedly steeper power curves in $k_{\text{prop}}$ and faster improvement with increasing $n$. Even in small samples (e.g., $n{=}20$–$30$), ED already performs favorably, frequently achieving the highest power among all methods, even when overall detection capability is constrained by sample size. Overall, the figure demonstrates a monotonic and robust gain profile for ED over FL, ECP, and PELT across low-to-moderate variance regimes, with ED providing the best early detection as $k_{\text{prop}}$ increases.\\

Under larger variance shifts (e.g., $\sigma^2_{\text{post}}=3$), ED’s superiority is even more evident; for example, by $n{=}100$ it regularly attains high power at mid-to-late change points (e.g., $k_{\text{prop}}{\ge}0.3$), and by $n{=}200$ ED approaches or exceeds $0.9$ across the grid, surpassing the next best competitor by a wide margin. Under these challenging conditions, ECP and PELT become competitive but remain distinctly below ED, while FL is uniformly the weakest. The dependence on the change-point location remains strong, and all methods benefit as $k_{\text{prop}}$ moves rightward. However, ED’s slope with respect to $k_{\text{prop}}$ is steeper and yields larger relative gains when more post-change samples are available. Across all panels, ED performs better in small $n$, and its lead widens as the variance contrast grows and the sample size increases. These results indicate that ED is the most reliable and powerful detector of exponential scale changes, and offers superior performance both when the underlying effect is weak and when the shift magnitude is substantial. Across all scenarios, ED exhibits a smooth and stable improvement as the sample size increases, which is advantageous for empirical applications. It should be noted that, in terms of computational efficiency, the proposed ED method is comparable to ECP and exhibits the expected quadratic scaling with sample size due to pairwise distance computations. On a standard Intel i7 processor, the average runtime for $n=200$ was approximately $1.05$ seconds for ED, compared to $1.18$ seconds for ECP, while PELT and FL required longer runtimes of about $1.75$ seconds and $2.30$ seconds, respectively. Overall, ED achieves the smallest runtime among the methods considered, with ECP as the second fastest.\\

\subsection{Localization Accuracy}
In addition to detection power, accurately locating the change point is essential. We evaluate localization accuracy using the scaled absolute error $|\hat{k}-k^*|/n$, where smaller values indicate more precise estimation. Table $\ref{loc1}$ presents the average localization error for the Normal mean-shift setting with $n=100$, comparing the proposed ED method with FL, ECP, and PELT. This comparison highlights how well ED localizes the change point relative to these competing approaches.

{\setlength{\tabcolsep}{1em}
\begin{table}[H]
\caption{Average localization error $|\hat{k} - k^*|/n$ for $N(0,1)$ data with mean shift $\delta$.}
\centering
\begin{tabular}{lcccc}
\hline
Method & $\delta=0.5$ & $\delta=1.0$ & $\delta=1.5$ & $\delta=2.0$ \\
\hline
ED  & 0.11 & 0.04 & 0.02 & 0.01 \\
FL & 0.14 & 0.06 & 0.03 & 0.01 \\
ECP & 0.11 & 0.04 & 0.02 & 0.01 \\
PELT & 0.18 & 0.09 & 0.04 & 0.02 \\
\hline
\end{tabular}
\label{loc1}
\end{table}}

Table $\ref{loc1}$ reports the average scaled localization error $|\hat{k}-k^*|/n$ for the Normal mean-shift setting with $n=100$, where values closer to $0$ indicate better precision in estimating the change-point location. The proposed method, ED, achieves consistently high localization accuracy across all signal strengths, with error decreasing from $0.11$ at $\delta=0.5$ to $0.01$ at $\delta=2.0$. For weak and moderate shifts ($\delta=0.5,1.0,1.5$), ED matches the best-performing competitor ECP exactly, indicating comparable precision in locating the true change point when the signal is subtle. ED also outperforms FL and PELT at these smaller $\delta$ values, with FL showing slightly larger errors (e.g., $0.14$ vs.\ $0.11$ at $\delta=0.5$) and PELT exhibiting the largest errors overall (e.g., $0.18$ at $\delta=0.5$). As $\delta$ increases to $2.0$, all methods improve, but ED remains among the most accurate with error $0.01$, tied with FL and ECP and better than PELT ($0.02$). Overall, these results suggest that ED provides reliable and precise localization, performing comparably to ECP and offering clear advantages over FL and especially PELT when the mean shift is small.

\section{Application}

\sloppy In this section, to demonstrate the change-point detection procedure on real data, we apply the proposed method to breast cancer data. Genomic DNA copy number alterations are critical genetic events in the development and progression of human cancers. The dataset includes DNA copy number changes across 6,691 human genes in 44 predominantly advanced primary breast tumors and 10 breast cancer cell lines. This dataset is publicly available\footnote{\url{https://www.ncbi.nlm.nih.gov/geo/query/acc.cgi?acc=GSE3281}.} and has been used in previous studies for change-point detection. For example, \cite{tib2008} applied the FL to the same dataset to smooth noisy copy number measurements and detect chromosomal aberrations. \\

To detect all change points in each chromosome, we apply the proposed method along with the binary segmentation approach. While the FL approach focuses primarily on enforcing local smoothness and penalizing abrupt transitions, the ED method captures global distributional shifts between adjacent genomic regions. This nonparametric property enables ED to detect subtle yet statistically meaningful changes that may not be accompanied by sharp mean-level differences. The results are presented in Figure \ref{fig_app}, where the detected change points are marked with red dashed lines. Across the examined chromosomes, the ED-based procedure identifies a richer and more refined set of genomic alterations compared to the FL analysis. Assumptions~\ref{A:sampling}–\ref{A:delta} are stated for independent (piecewise i.i.d.) observations. In array CGH data, however, measurements along a chromosome can exhibit local serial dependence due to spatial proximity and preprocessing. Such dependence can inflate scan statistics and may lead any i.i.d.-calibrated procedure to report additional breakpoints that reflect correlation structure rather than genuine copy-number alterations. To mitigate this in practice, we recommend calibrating the global statistic $T_n$ using a block-resampling scheme that preserves short-range dependence. One simple approach is a circular block permutation: fix a block length $M$ (e.g., chosen from the empirical autocorrelation range;
in our experiments we used $M=\lceil \sqrt{n}\rceil$ for a segment of length $n$).
Partition $\{1,\dots,n\}$ into contiguous blocks
$B_1=\{1,\dots,M\}$, $B_2=\{M+1,\dots,2M\}$, \dots\.
For each permutation replicate, draw a uniform random permutation of the blocks,
concatenate the permuted blocks to form a block-permuted sequence, and recompute
$T_n$  on the permuted sequence to obtain a block-based critical
value $c^{\mathrm{block}}_{\alpha,n}$. For the analysis presented in Figure \ref{fig_app} and Table \ref{app}, we employed this circular block permutation strategy with block size $M \approx \lceil \sqrt{n}\rceil$. The reported change points are therefore significant at the $\alpha=0.05$ level under a null hypothesis that accounts for local serial dependence, ensuring that the detected events reflect structural alterations rather than spatial correlation. \\

When comparing these results with those of \cite{tib2008}, the proposed ED-based method detects multiple breakpoints that likely correspond to biologically relevant copy number changes in chromosomes \#3, \#6, \#8, and \#19. In contrast, both methods agree that chromosome \#15 shows no evidence of structural changes, which reinforces the consistency of ED in distinguishing stable genomic regions. Furthermore, the number of changes identified in chromosomes \#5, \#7, \#9, \#10, and \#13 is roughly similar to those found by \cite{tib2008}. To evaluate the proposed method against modern segmentation techniques, we additionally compared our results with those obtained using WBS. Although WBS is highly sensitive to short segments, our analysis of the breast cancer data indicates that it tends to identify numerous short-range intervals in high-variance regions (e.g., within Chromosome \#8). In contrast, the ED method provides a more parsimonious segmentation that recovers the major structural alterations identified by WBS but is less prone to over-segmentation driven by local spatial artifacts. This suggests that the energy distance criterion offers a robust alternative for genomic data where distinguishing genuine copy number changes from background noise is critical. Table~\ref{app} reports the detected change points. For each chromosome, probes are ordered by physical position, and the
reported locations refer to the probe indices in the ordering used in our preprocessing and Figure~\ref{fig_app}. To avoid any
ambiguity, we recommend also reporting the corresponding probe IDs and/or physical genomic coordinates
(base-pair positions) alongside the index values. The additional events identified by ED, particularly in the middle and distal regions of chromosomes \#3 and \#19, may reflect smaller, however, meaningful distributional shifts in copy number intensity that the penalized smoothing framework of the FL tends to oversmooth. A major advantage of the ED approach is its computational efficiency and interpretability. Unlike optimization-based methods that require tuning multiple penalty parameters, the ED test relies on a unified distance-based criterion and adapts naturally to varying signal strengths and data scales.

{\setlength{\tabcolsep}{1 em}
\begin{table}[H]
\footnotesize
\caption{Change Point in chromosomes}
\centering
\begin{tabular}{cl}
\hline
Chromosomes Number & Changepoint Locations  \\  \hline
3 & 1569, 1588, 2085, 2756, 2808, 2839, 3052, 3067, 3092,  3171, 3367, 3497, 3507, 4972,\\
  & 4988, 4998, 5009, 5048, 5067, 5214, 6644, 6656 \\ \hline
5 & 529, 541, 727, 797, 816, 4997, 5014, 5067, 5081, 5349, 5370, 5407, 6614, 6644, 6657 \\ \hline
6 & 2442, 2455, 2485, 2547, 2598, 2667, 6565, 6585, 6615, 6716 \\ \hline
7 & 1795, 1807, 2437, 2460, 3245, 3262, 6074, 6222, 6284 \\ \hline
8 & 849, 1222, 1913, 1926, 2543, 2766, 2809, 2846, 2862, 3067, 3177, 3218,  3232, 3246, \\
  & 3269, 3327, 3565, 4333,  4356, 4367, 4577, 4599, 5550 \\ \hline
9 & 344, 6643, 6653 \\ \hline
10 & 1842, 1867, 3532, 4840, 4999, 5180, 5215 \\ \hline
13 & 2031, 2042, 2553 \\ \hline
15 & None \\ \hline
19 & 1360, 2240, 2436, 2454, 2469, 2774, 2808, 2826, 4272,  4358, 4697, 4850, 4933, 4945,\\
   & 5067, 5081, 5820, 5831, 5872, 5886, 6218, 6228, 6250, 6336 \\ \hline
\end{tabular}
\label{app}
\end{table}}

\begin{figure}[H]
  \centering
  \includegraphics[width=1\textwidth]{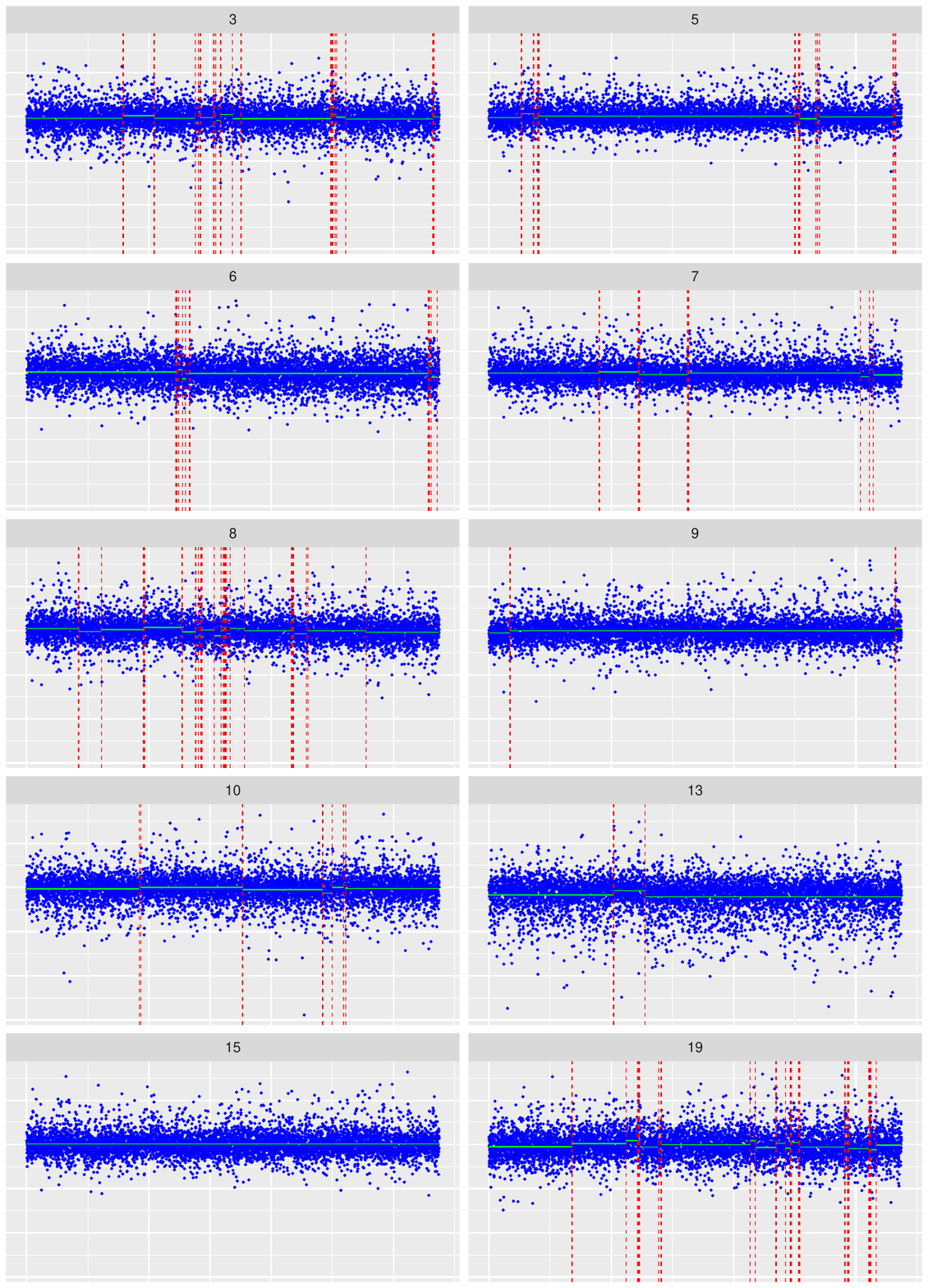}
  \caption{Array CGH profile of 10 chromosomes of breast cancer cell line MDA157}
  \label{fig_app}
\end{figure}

\section{Conclusion}
In this paper, we introduced a change-point detection procedure based on the energy distance and established its
large-sample behavior under both the null and alternative hypotheses. The proposed ED scan statistic is built from a
weighted two-sample $U$-statistic and equipped with an explicit studentization, yielding an asymptotically pivotal
quantity for interior split points. This construction leads to a transparent global test based on maximizing the
studentized statistic over candidate splits. Because the resulting scan involves dependent statistics across $k$, we
emphasized calibration for the maximum-type functional, and we adopted a permutation-based approach that delivers
reliable finite-sample performance while remaining asymptotically valid under $H_0$. We conducted an extensive simulation study under normal, skew-normal, and exponential baselines, considering a
range of change magnitudes and locations. Across these settings, ED provides stable Type~I error control and
demonstrates strong power gains as the sample size increases and as the change point moves away from the boundary.
In challenging regimes, including small to moderate samples, boundary change points, and subtle contrasts in mean,
scale, or shape, ED is routinely competitive with, and often superior to, Fused Lasso (FL), E-Divisive (ECP), and
PELT. When the distributional contrast is large, all methods approach unit power, and ED remains at least as effective
as the strongest competitor, frequently attaining high power at smaller sample sizes. In addition to detection power,
the empirical localization results indicate that ED provides accurate estimation of the change-point location across
signal strengths. The practical utility of the method was further illustrated with cDNA microarray CGH data. Combined with binary
segmentation, ED detects a rich set of structural changes and, for several chromosomes, identifies additional candidate
breakpoints relative to previously reported analyses. These additional signals may correspond to subtle distributional
changes; at the same time, they motivate careful interpretation in the presence of possible serial dependence along the
chromosome. Overall, the theoretical and empirical evidence supports ED as a robust, powerful, and computationally
attractive default for single-change detection. Although the present theory targets a single change point under independence assumptions, the procedure can be
embedded directly into standard segmentation frameworks to address multiple changes in practice. Several extensions
are natural directions for future work. First, it would be valuable to develop a comprehensive theory for multi-change
settings that accounts explicitly for temporal dependence. Second, extending the methodology to high-dimensional or
structured observations (including feature screening or aggregation) would broaden applicability to modern genomic
and imaging data. Finally, sequential variants of the ED scan statistic could enable online monitoring and real-time
change detection with principled error control.

\section*{\textsc{References}}
\renewcommand{\section}[2]{}
\bibliographystyle{apacite}
\setcitestyle{authoryear, open={((},close={))}}
\bibliography{References}

\end{document}